\documentclass[a4paper,pra,twocolumn,superscriptaddress,groupedaddress]{revtex4}
\usepackage{ae}
\usepackage{physics}
\usepackage[T1]{fontenc}
\usepackage[ansinew]{inputenc}
\usepackage{amsmath}
\usepackage{amssymb}
\usepackage{amsthm}
\usepackage[caption=false]{subfig}
\usepackage{multirow}
\usepackage{array}
\usepackage{wrapfig}
\usepackage{makecell}
\usepackage{xparse}
\usepackage{dcolumn}
\usepackage[mathscr]{euscript}
\usepackage{pdfpages}
\usepackage[]{graphicx}

\usepackage{color}
\usepackage[colorlinks]{hyperref}
\usepackage{lscape}
\theoremstyle{remark}
\newtheorem{theorem}{Theorem}
\newtheorem{lemma}{Lemma}
\newtheorem{corollary}{Corollary}
\newtheorem{example}{Example}
\graphicspath{ {./images1/} }

\usepackage{tikz}

\begin{document}
    \title{Exploring strong locality : Quantum state discrimination regime and beyond}
    
    \author{Subrata Bera}
    \email{98subratabera@gmail.com}
    \affiliation{Department of Applied Mathematics, University of Calcutta, 92, A.P.C. Road, Kolkata- 700009, India}
    
    \author{Atanu Bhunia}
    \email{atanu.bhunia31@gmail.com}
    \affiliation{Department of Applied Mathematics, University of Calcutta, 92, A.P.C. Road, Kolkata- 700009, India}
    
    \author{Indranil Biswas}
    \email{indranilbiswas74@gmail.com}
    \affiliation{Department of Applied Mathematics, University of Calcutta, 92, A.P.C. Road, Kolkata- 700009, India}
    
    \author{Indrani Chattopadhyay}
    \email{icappmath@caluniv.ac.in}
    \affiliation{Department of Applied Mathematics, University of Calcutta, 92, A.P.C. Road, Kolkata- 700009, India}
    
    \author{Debasis Sarkar}
    \email{dsarkar1x@gmail.com,dsappmath@caluniv.ac.in}
    \affiliation{Department of Applied Mathematics, University of Calcutta, 92, A.P.C. Road, Kolkata- 700009, India}	

 \begin{abstract}
Based on the conviction of switching information from locally accessible to locally hidden environs, the concept of hidden nonlocality activation has recently been highlighted by Bandyopadhyay et al. in [Phys. Rev. A 104, L050201 (2021)]. They have demonstrated that a certain locally distinguishable set of pure quantum states can be transformed into a locally indistinguishable set with certainty through orthogonality preserving local measurements (OPLMs). This transformation makes the set locally inaccessible, despite being locally accessible before. This phenomenon is defined as the activation of hidden nonlocality. In this paper, we present two classes of locally distinguishable sets within $(2m+1) \otimes 2 \otimes (2m+1)$ systems. One class reveals nonlocality through local operations, whereas the other requires joint measurements for it. As the later class depends on nonlocal operations to exhibit nonlocality, it arguably has a lower degree of nonlocality, and accordingly, can be considered as more local compared to the first class. This analysis exhibits a stronger manifestation of locality by elucidating the nuanced interplay between these distinct local phenomena within the framework of quantum state discrimination. Furthermore, we also explore their significant applications in the context of data hiding. Additionally, we introduce the concept of \emph{``strong local"} set and compare it with various activatable sets, highlighting differences in terms of locality.
\end{abstract} 
	
    \date{\today}
    \maketitle

     \section{INTRODUCTION}
 
Apart from the Bell nonlocality, there are other manifestations of nonlocality that have fascinated researchers in recent decades. In fact, the local indistinguishability of some sets of multipartite orthogonal quantum states has been comprehensively used to embellish the phenomenon of quantum nonlocality \cite{BennettPB1999,BennettUPB1999,bennett1996,popescu2001,xin2008,Walgate2000,Virmani,Ghosh2001,Groisman,Walgate2002,Divincinzo,Horodecki2003,Fan2004,Ghosh2004,Nathanson2005,Watrous2005,Niset2006,Ye2007,Fan2007,Runyo2007,somsubhro2009,Feng2009,Runyo2010,Yu2012,Yang2013,Zhang2014,somsubhro2009(1),somsubhro2010,yu2014,somsubhro2014,somsubhro2016}. In this scenario, a state is secretly chosen from a known set of orthogonal quantum states. The goal is to identify the state by performing local operations and classical communications(LOCC). If they can identify each state with certainty, the set is pronounced to be locally distinguishable. Otherwise, the set is called locally indistinguishable. Local indistinguishability of quantum states is a key concept to understand the restrictions of LOCC. In this context, it might seem intuitive to assume that entanglement is necessary and sufficient for local indistinguishability due to complex nature of quantum correlations. However, researcher has demonstrated that entanglement is not necessary for the local indistinguishability of quantum states \cite{BennettPB1999,BennettUPB1999,Zhang2015,Wang2015,Chen2015,Yang2015,Zhang2016,Xu2016(2),Zhang2016(1),Xu2016(1),bhunia2023,bhunia2024,biswas2023, Halder2019strong nonlocality,Halder2019peres set,Xzhang2017,Xu2017,Wang2017,Cohen2008,Zhang2019,somsubhro2018,zhang2018,Halder2018,Yuan2020,Rout2019,bhunia2020,bhunia2022}. Bennett et. al. \cite{BennettPB1999} was the first in 1999 to demonstrate a set of nine pure orthogonal product states in $3\otimes3$ system that cannot be perfectly distinguished using LOCC, revealing the phenomenon of nonlocality without entanglement in the sense that they are locally indistinguishable. 
On the contrary, in 2000 Walgate et. al.\cite{Walgate2000} showed that any two multipartite orthogonal pure states whether product or entangled can be perfectly distinguished by LOCC. This result indicates that entanglement is not also sufficient for local indistinguishability. Local indistinguishability of quantum states has been practically applied in quantum cryptography primitives such as data hiding, secret sharing, etc.,\cite{ terhaldatahiding, haydendatahiding, winterdatahiding, wehner2020, chaves2020, lamidatahiding, plahiding}. Consequently, locally indistinguishable sets can be considered as resources in quantum information processing. A relevant question in this context was explored by Bandyopadhyay and Halder in 2021 \cite{BandyopadhyayHalderActivation2021}, where they investigated whether it is possible to locally transform a local set of quantum states into a nonlocal one that can be used for data hiding. They provided a positive answer, demonstrating that certain local sets of orthogonal pure states can be deterministically transformed into locally indistinguishable sets through orthogonality preserving local measurements(OPLMs). This phenomenon is called activation of hidden nonlocality. Later, in 2022, Li and Zheng further extended these findings by introducing a locally distinguishable set of bipartite product states that can be precisely converted to a locally irreducible (and hence indistinguishable) set through OPLMs\cite{Li2022}. In the same year, Ghosh et al. provided a scenario to activate strong nonlocality from a local set where the sets of postmeasurement states are not only locally indistinguishable but also remain nonlocal even after performing joint measurements \cite{GhoshStrongActivation2022}. Here, it is important to note that when we refer to a set as \emph{``local"}, we specifically mean that it is locally distinguishable. Likewise, \emph{``nonlocal"} means that the set is locally indistinguishable. This usage occurs frequently throughout the paper. It is crucial not to confuse these with the terms as used in the context of Bell nonlocality, where they carry entirely different meanings.   

The activation of hidden nonlocality in multipartite scenarios has not yet been fully explored, and many important questions remain unanswered. For instance, could there be a locally distinguishable set whose nonlocality cannot be activated through local measurements but instead requires joint measurements? In such cases, revealing the hidden nonlocality necessitates nonlocal measurements, in contrast to the previously mentioned locally activable sets. Accordingly, it can be claimed that the degree of nonlocality of such a set is lower than that of the mentioned locally activable sets, and therefore, it can be considered as more local in comparison. This paper aims to address this question by analyzing the activation of hidden nonlocality in various distinguishable sets under local projective measurements and classical communications.

Local projective measurements and classical communications(LPCC) are strict restriction of LOCC. Consequently, any task that can be achieved using LPCC is also feasible under LOCC; however, the converse is not always true \cite{BennettUPB1999,Cohen2007,Xin2008}. Understanding and probing the gap between LOCC and LPCC provides a significant challenge, offering valuable insights into the limitations and potentials of the state discrimination tasks under these differing operational frameworks in quantum information theory. In this paper, we first construct two LPCC distinguishable sets, $S_1$ and $S_2$ within a separable complex Hilbert space $\mathcal{H}^3_{A}\otimes\mathcal{H}^2_{B}\otimes\mathcal{H}^3_{C}$. For $S_1$, we show that the activation of nonlocality is possible through OPLMs under LPCC, whereas this is impossible for the set $S_2$. Interestingly, however, activation is achieved for $S_2$ when Bob($B$) and Charlie($C$) are allowed to perform joint projection-valued measures(PVMs). Under the LPCC scenario, this result indicates that $S_2$ exhibits a lower degree of nonlocality compared $S_1$. Conversely, $S_2$ can be considered to have a higher degree of locality than $S_1$. In other words, $S_2$ is more local than $S_1$. This motivates our exploration into comparing the degrees of locality among different classes of quantum states based on the activation of hidden nonlocality. After that, we have generalized this idea in $(2m+1) \otimes 2 \otimes (2m+1)$ system (the sets $S_{1,m}$ and $S_{2,m}$ associated with $S_1$ and $S_2$ respectively). Finally, we end up by constructing a local orthogonal set $S$ in $8 \otimes 8 \otimes 8$ system, whose nonlocality cannot be activated through OPLMs under LPCC, but instead requires joint measurements by any pair of parties. We begin with some preliminary notions which we require in the subsequent discussions.

\section{Preliminaries}
A measurement on a $d$-dimensional quantum
system can be expressed as a set of \emph{positive operator-valued measure (POVM)} elements $\{M_l\}_l$. These elements are the positive semi-definite Hermitian matrices that satisfy the completeness relation $\sum_l M_l = \mathbb{I}_d$, where $\mathbb{I}_d$ is
the identity matrix of order d.

\emph{Definition 1}--- If all the POVM elements of a measurement set-up associated with a state discrimination task for a given set of quantum states are proportional to the identity matrix, then such a measurement is not useful for extracting information for this task and is called a \emph{trivial measurement}. Conversely, if not all POVM elements of a measurement are proportional to the identity matrix then the measurement is said to be a \emph{nontrivial measurement}\cite{Halder2018}.

A necessary and sufficient condition for distinguishing a set of quantum states is that the states are mutually orthogonal. Consequently, any discrimination protocol chosen to distinguish these states involves only POVM elements that preserve their orthogonality (it might also eliminate one or more states). Such measurements are said to be orthogonality preserving measurements (OPMs)\cite{Halder2018}. In particular, when these measurements act locally, they are said to be orthogonality preserving local measurements (OPLMs)\cite{Walgate2002, BandyopadhyayHalderActivation2021}. 

The concept of local activation begins with a locally distinguishable set of pure quantum states. The objective is to identify OPLMs in such a way that every set of postmeasurement states becomes locally indistinguishable. This phenomenon is known as the \emph{activation of hidden nonlocality} \cite{BandyopadhyayHalderActivation2021}. The authors argued that this activation is \emph{genuine} for the sets free from local redundancy. 

\emph{Definition 2}---A set of mutually orthogonal pure states is said to be \emph{locally redundant} if and only if the orthogonality among its states is preserved even after discarding one or more subsystems. Such sets exhibit trivial activation phenomena, as extensively described in \cite{GuptaGhoshHierarchicalActivation2023}. Note that:

(a) Sets containing locally irredundant subsets automatically avoid local redundancy.

(b) Local redundancy is impossible if every local dimension of the corresponding system is prime.

A particularly interesting subclass of POVM measurements happens to be \emph{projection-valued measure(PVM)}, which is a set of orthogonal projectors $\{P_l\}$ that sum to the identity matrix: $\sum_l P_l = \mathbb{I}$ and $P_iP_j=\delta_{ij}P_i$.

\emph{Definition 3}--- In a LOCC protocol when the local measurements are restricted to PVMs instead of POVMs, such a protocol is commonly known to be \emph{local projective measurements and classical communications (LPCC)}.

\emph{Definition 4}--- A set of quantum states is said to be \emph{locally irreducible} \cite{Halder2019strong nonlocality} if it is not possible to eliminate one or more states under OPLMs. A sufficient condition for the irreducibility of a set is that there exist no nontrivial measurements as OPLMs. Local irreducibility sufficiently ensures local indistinguishability. However, the converse is not true.

For the sake of brevity, we have used the unnormalized states throughout the paper.

\section{activation of nonlocality: from locally accessible to locally inaccessible information}

In the realm of quantum information theory, we encounter two distinct classes of locally distinguishable sets within the $n$-partite quantum system, say $\mathscr{S}_1$ and $\mathscr{S}_2$. None of these sets can be directly employed to securely transmit classical information as the information encoded in a locally distinguishable set is always locally accessible. We consider $\mathscr{S}_1$ to be a set where nonlocality(in the sense of local indistinguishability) can be activated through local operations. In the context of data hiding and secret sharing, such phenomena can effectively turn a resourceless set into a resourceful one. It can be shown that such a set can be converted to a locally indistinguishable set through OPLMs with certainty. This transformation allows us to hide information locally within $\mathscr{S}_1$. In other words, the information encoded in $\mathscr{S}_1$ which was locally accessible initially, after transformation, can not be accessed completely by the local observers; part of it will always remain hidden. Naturally, the realization of a stronger form of nonlocality activation is an immediate question to ponder. This is precisely what Ghosh et al. explained in their paper \cite{GhoshStrongActivation2022}. They have presented a stronger version of nonlocality activation related to the impossibility of state elimination, beyond the nonlocal aspects of state discrimination. Specifically, they demonstrated some classes of genuinely activable sets that can be transformed into locally irreducible sets through OPLMs. This prompts a query about the existence of the strongest form of genuinely activable nonlocality. For that, they further present a genuinely activable set (here we denote it as $\mathscr{U}_1$), where the postmeasurement states are not only locally irreducible but also remain irreducible in every partition. Consequently, in this configuration, the parties within the corresponding quantum system can transform $\mathscr{U}_1$ to a strong nonlocal set through OPLMs.

On the other hand, $\mathscr{S}_2$ represents such sets where genuine activation is impossible through OPLMs. This means that, when the parties are bound to act locally, the outcomes consistently exhibit local sets. Consequently, the classical information encoded in the states of such a set always remains accessible to the local observers. Thus, $\mathscr{S}_2$ appears to be unsuitable for applications in secret sharing and data hiding. This characteristic places $\mathscr{S}_2$ on the more local end of the spectrum compared to $\mathscr{S}_1$, as the hidden nonlocality of $\mathscr{S}_2$ cannot be revealed through local operations as effectively as with $\mathscr{S}_1$. An important question may arise here: Could $\mathscr{S}_2$ still contain hidden nonlocality that has yet to be unveiled? If so, understanding the nature of such nonlocality and the associated methods of accessing it becomes crucial. To demonstrate this, we present an example of a local set whose nonlocality cannot be activated by local observers but can be activated through joint measurements performed by the two observers. In this context, we will always avoid using global measurements, as any set of orthogonal quantum states is globally distinguishable. Therefore, global measurements are incapable of activating nonlocality in any set.

Further exploration reveals an intriguing scenario: some local sets remain impervious to activating nonlocality even after using joint measurements. Consequently, we can assert that such sets can never serve as viable quantum resources for data hiding and secret sharing. From the activation perspective, we can classify these sets as \emph{strong local sets}, denoted as $\mathscr{U}_2$. Numerous examples of local sets exhibit this distinctive feature:

(1) It is well known that any pair of orthogonal pure states is locally distinguishable\cite{Walgate2000} and therefore, remains locally accessible under OPLMs. Consequently, any set containing only two pure orthogonal states is strongly local, regardless of the number of local observers or the dimension of the corresponding quantum system. 

(2) In \cite{BandyopadhyayHalderActivation2021}, Bandyopadhyay and Halder proved that any local set in $2 \otimes 2$ cannot be activated under OPLMs, further exemplifying the concept of strong local sets. 

(3) A set of three pure product states is strong local. This can be shown through the following:

(a)\emph{A set of three orthogonal pure product states is distinguishable under LPCC(LOCC ).} Consider an n-party quantum system with three orthogonal product states $\ket{\phi_1}=\otimes_{i=1}^{n}\ket{\alpha_i}$, $\ket{\phi_2}=\otimes_{i=1}^{n}\ket{\beta_i}$ and $\ket{\phi_3}=\otimes_{i=1}^{n}\ket{\gamma_i}$, where $\braket{\alpha_i}{\alpha_i}=\braket{\beta_i}{\beta_i}=\braket{\gamma_i}{\gamma_i}=1$, for $i=1,2,\ldots, n$ and $\braket{\phi_l}{\phi_k}=0$, for $l\neq k$. 

There exist atleast one $j$ such that $\braket{\alpha_j}{\beta_j}=0$. We can write $\ket{\gamma_j}=a\ket{\alpha_j}+b\ket{\beta_j}+c\ket{\eta_j}$, where $\ket{\eta_j}\in \text{span}\{\alpha_j,\beta_j\}^{\perp}$ and $|a|^2+|b|^2+|c|^2=1$.

Suppose party-$j$ measures using two PVMs $M_0=\ketbra{\alpha_j}{\alpha_j}$ and $M_1=\mathbb{I}-M_0$.

If the outcome is $0$, the postmeasurement states will be

$\ket{\phi_1}$ and $\ket{\phi_3'}=\otimes_{i=1}^{j-1}\ket{\gamma_i}\otimes(a\ket{\alpha_j})\otimes_{i=j+1}^{n}\ket{\gamma_i}$

Therefore, 
\begin{equation}
\begin{array}{l}
\braket{\phi_1}{\phi_3'}\\

=a.\braket{\alpha_j}{\alpha_j}.\prod_{i=1,\neq j}^{n}\braket{\alpha_i}{\gamma_i}\\

=\braket{\alpha_j}{a\alpha_j+b\beta_j+c\eta_j}.\prod_{i=1,\neq j}^{n}\braket{\alpha_i}{\gamma_i}\\

=\braket{\alpha_j}{\gamma_j}.\prod_{i=1,\neq j}^{n}\braket{\alpha_i}{\gamma_i}\\
=\prod_{i=1}^{n}\braket{\alpha_i}{\gamma_i}\\
=\braket{\phi_1}{\phi_3}\\
=0
\end{array}
\end{equation}

Thus, $\ket{\phi_1}$ and $\ket{\phi_3'}$ are two orthogonal pure product states and hence distinguishable under LPCC(LOCC ).

Similarly, for the outcome $1$ we will also get two orthogonal pure product states which are distinguishable under LPCC(LOCC ).

(b)\emph{OPLM cannot create entanglement from product states.} Therefore every outcome under OPLM must lead to a distinguishable set.

Now our aim is to illustrate the entire concept within a one-dimensional framework, as depicted in FIG.[\ref{fig:Line}]. Our main objective is to arrange the aforementioned sets along the line $L_2$, where the degree of locality increases from left to right.

\begin{figure}[htp]
	\centering
	\includegraphics[scale=.42]{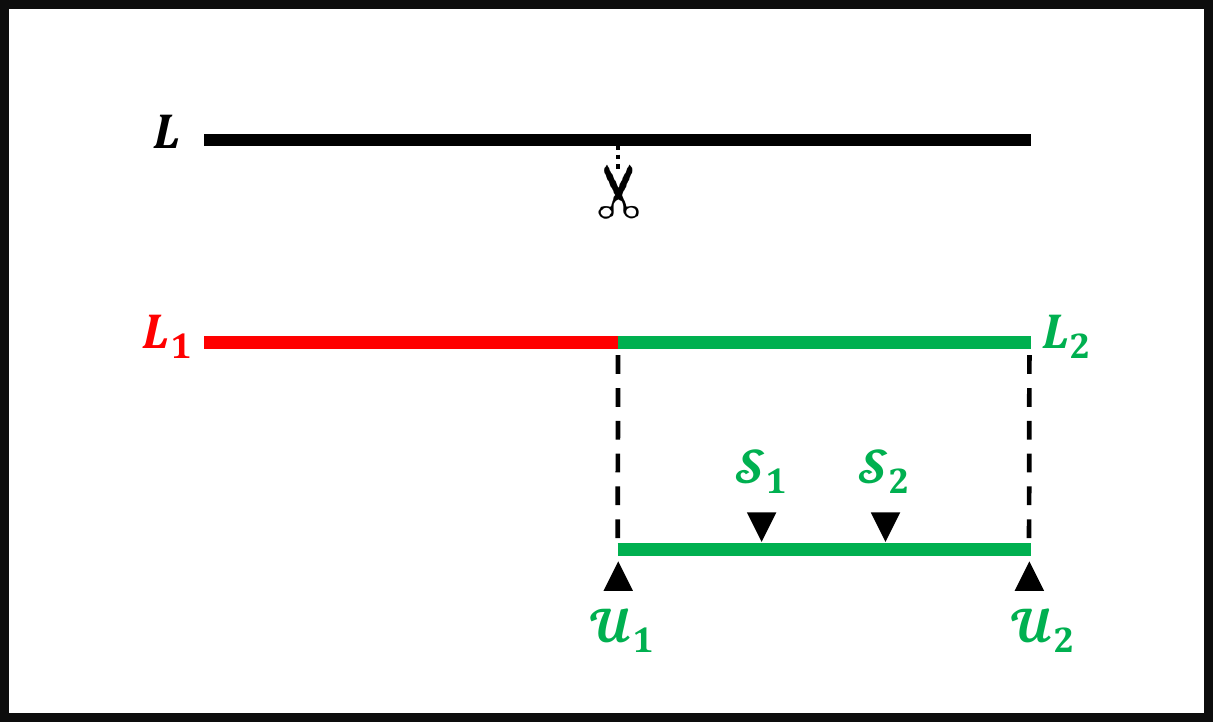}
	\caption{Imagine a line $L$ divided into two segments, $L_1$ and $L_2$. On this line, $L_1$ represents locally indistinguishable sets, while $L_2$ represents locally distinguishable sets. Within this framework, local sets like $\mathscr{U}_1$, which exhibit activation of strong nonlocality, are positioned at the leftmost part of $L_2$. On the other hand, strong local sets such as $\mathscr{U}_2$ occupy the rightmost part of $L_2$. It's important to note that the sets $\mathscr{S}_1$ and $\mathscr{S}_2$ may fall somewhere in between, but $\mathscr{S}_1$ must be to the left of $\mathscr{S}_2$.}
	\label{fig:Line}
\end{figure}

In this paper, we use PVMs as our optimal measurement tool because they provide many interesting results for building our conceptual framework. In contrast, when using the POVMs, we find a lack of compelling results to fully establish our hypothesis. Therefore, in the early stages of establishing our innovative idea, the use of this particular class of measurements proved more than adequate. Notably, we are discussing those quantum systems having at least one subsystem with dimension two. To support our whole work, we need to coin some interesting properties of the set of states from $2\otimes n$ quantum system.

\begin{lemma}---
	Any set of orthogonal product states in $2 \otimes n$ is completely distinguishable under LPCC.
\end{lemma}
		
\begin{proof}
   A set of states in $\mathbb{C}^{2 \otimes n}$ can be written as 
   
   \begin{equation}
   	\begin{array}{c}
   		\ket{\phi_{ij}}=\ket{\alpha_i}\ket{\eta_{ij}}\\
   		\ket{\psi_{ik}}=\ket{\alpha_i^{\perp}}\ket{\kappa_{ik}}
   	\end{array}
   \end{equation}
   
   for some $i$, $j$, $k$. For each $i$, $\braket{\eta_{ij}}{\eta_{il}}_{j \neq l}=\braket{\kappa_{ik}}{\kappa_{im}}_{k \neq m}=0$ and for $i\neq h, \ket{\alpha_i} \not\propto \ket{\alpha_h},\braket{\alpha_i}{\alpha_h}\neq 0$ and $\text{span}\left\{\ket{\eta_{ij}},\ket{\kappa_{ik}}\right\}_{j,k} \perp \text{span}\left\{\ket{\eta_{hl}},\ket{\kappa_{hm}}\right\}_{l,m}$.
   
   The measurement strategy begins with Bob measuring in $P^B_i$, where for each $i$, $P^B_i$ is a rank-$d_i$ projective measurement defined as: $P^B_i=\sum_{r=1}^{d_i}\ket{\theta_{ir}}\bra{\theta_{ir}}$, where all $\left\{\ket{\theta_{ir}}\right\}$ satisfy $\text{span}\left\{\left\{\ket{\theta_{ir}}\right\}_{r=1}^{d_i}:\braket{\theta_{il}}{\theta_{im}}_{l\neq m}=0\right\}=\text{span}\left\{\ket{\eta_{ij}},\ket{\kappa_{ik}}\right\}_{j,k}$.
   
   If the outcome is $\mu$, then the postmeasurement states will be 
   
   \begin{equation}
   	\begin{array}{c}
   		\ket{\phi_{\mu j}}=\ket{\alpha_\mu}\ket{\eta_{\mu j}}\\
   		\ket{\psi_{\mu k}}=\ket{\alpha_{\mu}^{\perp}}\ket{\kappa_{\mu k}}
   	\end{array}
   \end{equation}
   
   Next Alice will measure in two rank-1 projective measurements:
   
   \begin{equation*}
   	\begin{array}{c}
   		P^B_{\mu 1}=\ket{\alpha_\mu}\bra{\alpha_\mu}\\
   		P^B_{\mu 2}=\ket{\alpha_\mu^{\perp}}\bra{\alpha_\mu^{\perp}}
   	\end{array}
   \end{equation*}
   
   If the outcome is 1, the postmeasurement states will be
   
   \begin{equation}
   	\begin{array}{c}
   		\ket{\phi_{\mu j}}=\ket{\alpha_\mu}\ket{\eta_{\mu j}}
   	\end{array}
   \end{equation}
   
   Finally Bob will measure again, but this time in rank-1 projective measurements, denoted as $P^B_{\mu 1 j}=\ket{\eta_{\mu j}}\bra{\eta_{\mu j}}$, which distinguish the states $\ket{\phi_{\mu j}}$.
   
   Similarly, if the outcome is 2 after Alice measures, the postmeasurement states will be
   
   \begin{equation}
   	\begin{array}{c}
   		\ket{\psi_{\mu k}}=\ket{\alpha_\mu^{\perp}}\ket{\kappa_{\mu k}}
   	\end{array}
   \end{equation}
   
   and these will be distinguished by Bob using the rank-1 projective measurements $P^B_{\mu 2 k}=\ket{\kappa_{\mu k}}\bra{\kappa_{\mu k}}$.
\end{proof}	

Furthermore, for such a local set if an observer applies an OPLM, the postmeasurement states remain product states. This is because POVM(PVM) cannot transform a product state into an entangled state. Therefore the set remains local, making activation impossible.

\begin{corollary}---
	It is impossible to activate the nonlocality of any set of orthogonal product states in $\mathbb{C}^{n \otimes 2}$ under LPCC as well as under LOCC.
\end{corollary}

Notably, it is essential to recognize that any set of orthogonal product states in $\mathbb{C}^{n \otimes 2}$ serves as an example of a strong local set from the activation perspective as mentioned earlier.

Next, we find out some circumstances where the parties with local dimension two have no power to activate nonlocality.

\begin{theorem}---
	Consider a tripartite quantum system $\mathcal{H}^n_{A}\otimes\mathcal{H}^2_{B}\otimes\mathcal{H}^2_{C}$ where a set of orthogonal states, which are product states in $A|BC$ partition, is given. The set is distinguishable under LPCC. If either Bob($B$) or Charlie($C$) performs a nontrivial PVM, the activation of nonlocality becomes impossible and can only be activated by Alice($A$).
\end{theorem}

\begin{proof}
	Suppose Bob goes first with PVM elements $\{P_l\}$. Since the dimension of the party Bob is two and $\sum P_l=\mathbb{I}$, Bob can only provide rank-$1$ projective measurements $P_1$ and $P_2$ as non-trivial measurements. For each outcome, the postmeasurement states will be product states in $\mathbb{C}^{n \otimes 1 \otimes 2}$. Consequently, it is impossible to activate nonlocality from here. If instead of Bob going first, Charlie goes first and perform non-trivial PVM, the postmeasurement states likewise become the product states in $\mathbb{C}^{n \otimes 2 \otimes 1}$ and therefore activation of nonlocality will also not occur in this case. Hence, if anyone can activate nonlocality, it will only be Alice.
\end{proof}

It is evident that any product state in a multipartite quantum system remains a product state across all partitions. Therefore, the above result also holds for the sets consisting of product states.
\begin{corollary}---
Consider a tripartite quantum system $\mathcal{H}^n_{A}\otimes\mathcal{H}^2_{B}\otimes\mathcal{H}^2_{C}$ where a set of orthogonal product states is distinguishable under LPCC. If either Bob($B$) or Charlie($C$) performs a nontrivial PVM, the activation of nonlocality becomes impossible and can only be activated by Alice($A$).
\end{corollary}
In the following part, we present the sets previously referred to as $\mathscr{S}_1$. These sets are categorized as TYPE-I Activable sets.
\subsection{TYPE-I ACTIVABLE SET}
The information encoded in such sets is fully accessible initially to all involved parties. However, remarkably, one of them has the capability to hide the information completely.
\begin{example}
    $S_1 \subset \mathbb{C}^{3 \otimes 2 \otimes 3}:$
\end{example}

The following set $S_1$[\ref{S_1}] of nine pure states in $\mathbb{C}^{3 \otimes 2 \otimes 3}$ can be distinguished under LPCC[\ref{distinguishability_S_1}].

\begin{equation}
\label{S_1}
	\begin{array}{l}
		\ket{\phi^1_{1}} \equiv \ket{0}\ket{00+01+10-11}\\
		\ket{\phi^1_{2}} \equiv \ket{0}\ket{00-01-10-11}\\
		\ket{\phi^1_{3}} \equiv \ket{1}\ket{01-11}\\
		\ket{\phi^1_{4}} \equiv \ket{2}\ket{01+02+11-12}\\
		\ket{\phi^1_{5}} \equiv \ket{2}\ket{01-02-11-12}\\
		\ket{\phi^1_{6,7}} \equiv \ket{0 \pm 1}\ket{02-12}\\
		\ket{\phi^1_{8,9}} \equiv \ket{1 \pm 2}\ket{00-10}
	\end{array}
\end{equation}

\begin{figure}[htp]
	\centering
	\includegraphics[scale=.29]{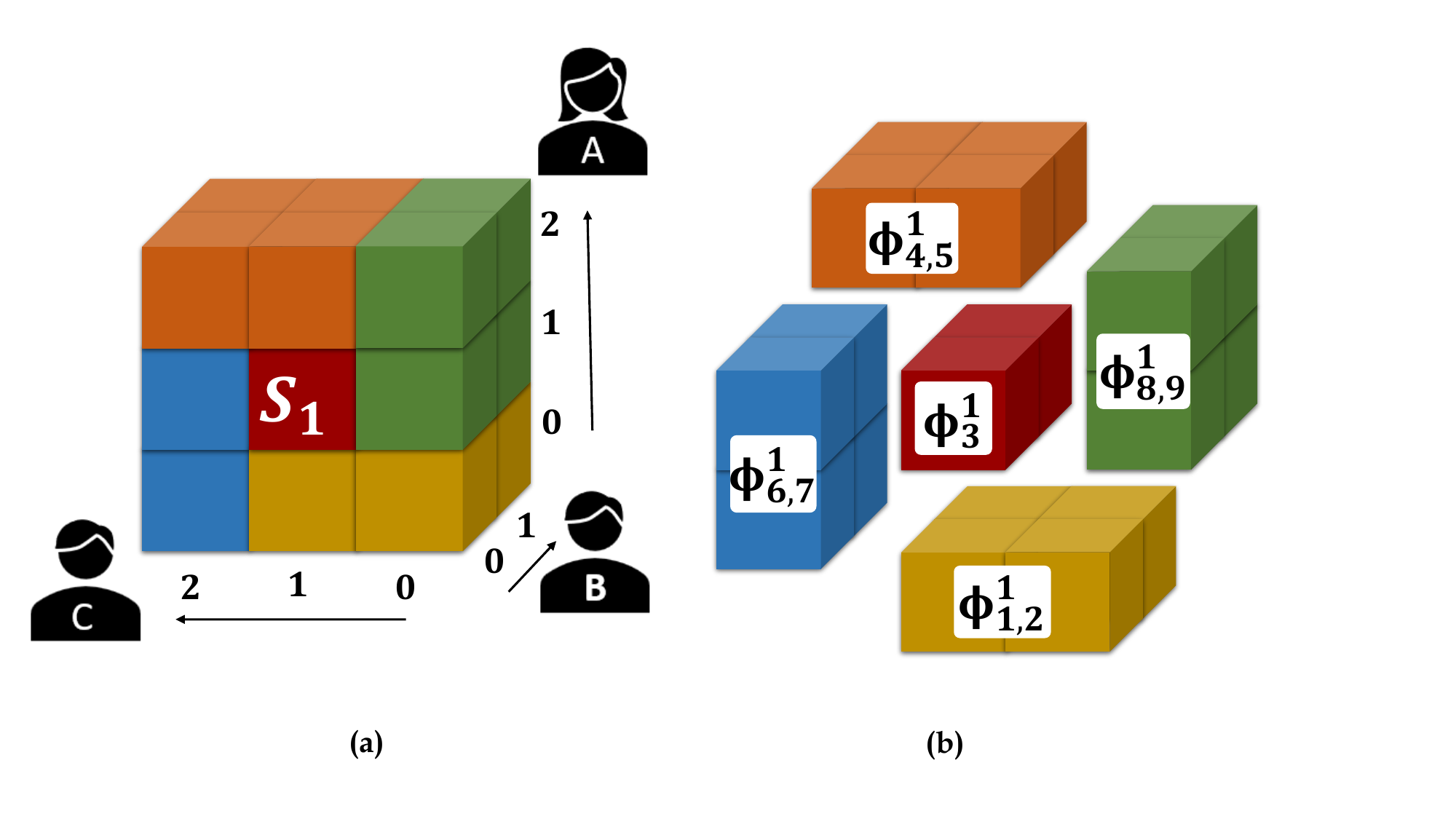}
	\caption{\textbf{(a)} The block structure for $\mathbb{C}^{3 \otimes 2 \otimes 3}$. An arbitrary sub-block made of $n$ basic blocks can contain at most $n$ mutually orthogonal states. \textbf{(b)} In this figure, all the sub-blocks are shown. The states of $S_1$ are placed in the corresponding sub-blocks.}
	\label{fig:S1}
\end{figure}

\begin{theorem}---
	It is possible to activate nonlocality in $S_1$ under LPCC.
\end{theorem}

Activating nonlocality in $S_1$ is very straightforward when Bob initiates the process and measures in $M_0=\ket{0}\bra{0}$ and $M_1=\ket{1}\bra{1}$. Upon examining each measurement outcome, we observe that the states in the A|BC partition are indistinguishable.  Further mathematical details are provided in Appendix-B[\ref{theorem_2}]. However, here we have included a pictorial representation[\ref{fig:S1Activation}] of the proof.

\begin{figure}[htp]
	\centering
	\includegraphics[scale=.3]{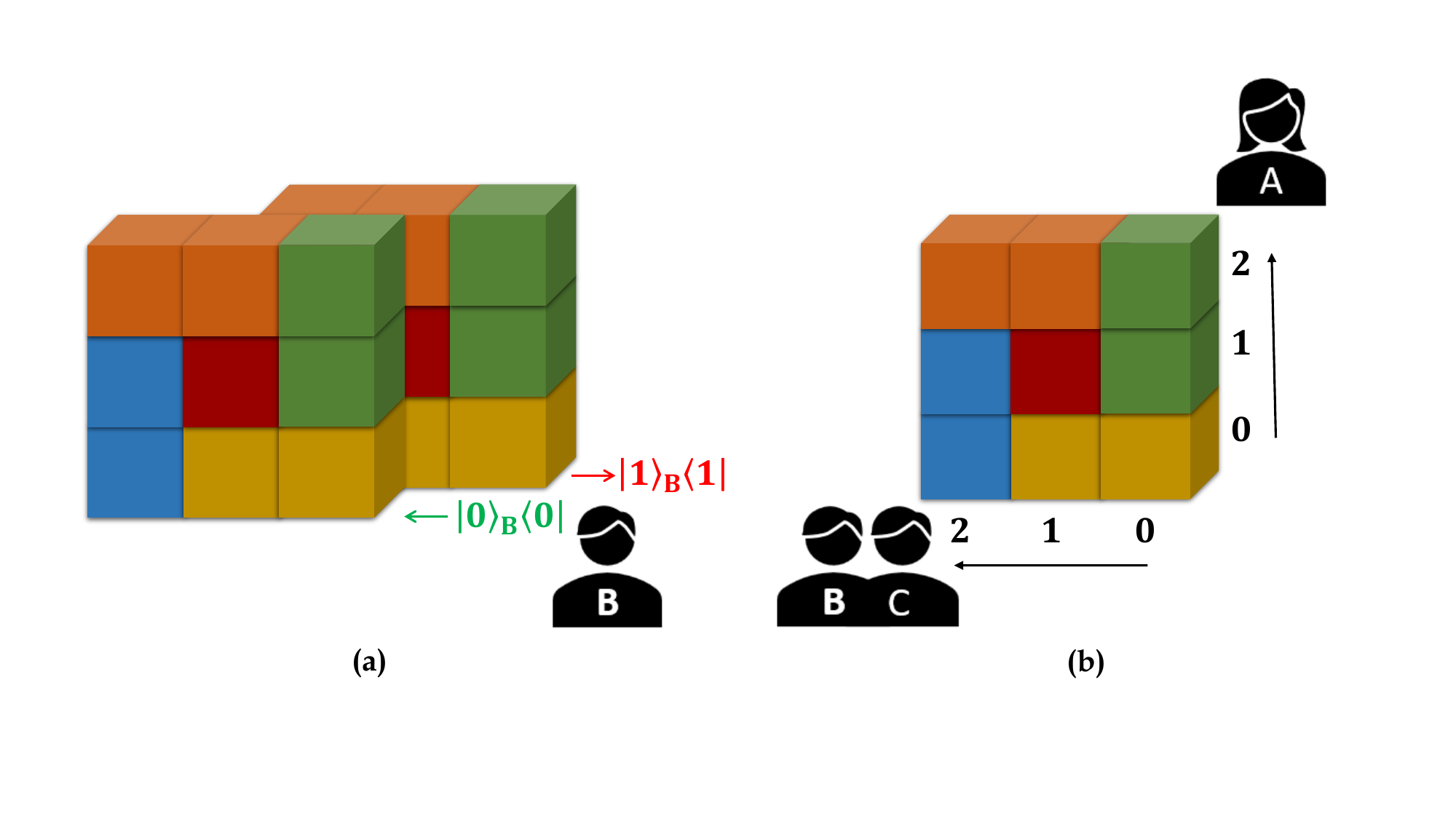}
	\caption{Pictorial representation of hidden nonlocality activation of $S_1$. Here Bob's measurement setup is $M_0=\ket{0}\bra{0}$ and $M_1=\ket{1}\bra{1}$. \textbf{(a)} The `$0$' outcome will isolate the front face of the block structure of $S_1$, while the `$1$' outcome will isolate its back face. In both cases, we end up on the nonlocal product basis proposed by Bennett et al. in 1999\cite{BennettPB1999}. \textbf{(b)} This figure represents the block structure of the aforementioned Bennett's set.}
	\label{fig:S1Activation}
\end{figure}

\begin{example}
    $S_{1,m} \subset \mathbb{C}^{(2m+1) \otimes 2 \otimes (2m+1)}:$
\end{example}

An $(2m+1) \otimes 2 \otimes (2m+1)$ generalization of $S_1$ is as follows:

\begin{equation}
	\begin{array}{l}
\vspace{5pt}\ket{\xi_1}=\ket{m}\ket{0-1}\ket{m}\\
		\text{For } i=0,1, \ldots, m-1, \text{ and } k=0, 1,\ldots, m-i-1,\vspace{3pt}\\
		\ket{\xi^1_{i,k}}=\ket{i}\left(\ket{0+1}\ket{i+2k}+\ket{0-1}\ket{i+2k+1}\right),\\
		\ket{\xi^2_{i,k}}=\ket{i}\left(\ket{0-1}\ket{i+2k}-\ket{0+1}\ket{i+2k+1}\right),\\
		\ket{\xi^3_{i,k}}=\ket{2m-i}\left(\ket{0+1}\ket{i+2k+1}+\ket{0-1}\ket{i+2k+2}\right),\\
		\ket{\xi^4_{i,k}}=\ket{2m-i}\left(\ket{0-1}\ket{i+2k+1}-\ket{0+1}\ket{i+2k+2}\right),\\
		\ket{\xi^5_{i,k}}=\left(\ket{i+2k+1}\pm\ket{i+2k+2}\right)\ket{0-1}\ket{i},\\
		\ket{\xi^6_{i,k}}=\left(\ket{i+2k}\pm\ket{i+2k+1}\right)\ket{0-1}\ket{m-i}
	\end{array}
\end{equation}

The block structure of $S_{1,m}$ is given in FIG.[\ref{fig:Generalization}] in the $A|BC$ partition. It is easy to check that for $m=1$ the above set coincides with our first example $S_1$, i.e., $S_{1,1} \equiv S_1$.\\

In the next part, we present some sets such as $\mathscr{S}_2$, which are classified as TYPE-II Activable sets.
\subsection{TYPE-II ACTIVABLE SET}
For such sets, the information is fully accessible to all involved parties initially, similar to type-I sets. However, in this scenario, no single party has the individual capability to hide the information completely. To achieve hiding of the information, at least two parties must perform a joint PVM.
\begin{example}
    $S_2 \subset \mathbb{C}^{3 \otimes 2 \otimes 3}:$
\end{example}

The following set $S_2$[\ref{S_2}] of nine pure states in $\mathbb{C}^{3 \otimes 2 \otimes 3}$ can be distinguished under LPCC[\ref{distinguishability_S_2}].

\begin{figure}[htp]
    \centering
    \includegraphics[scale=.29]{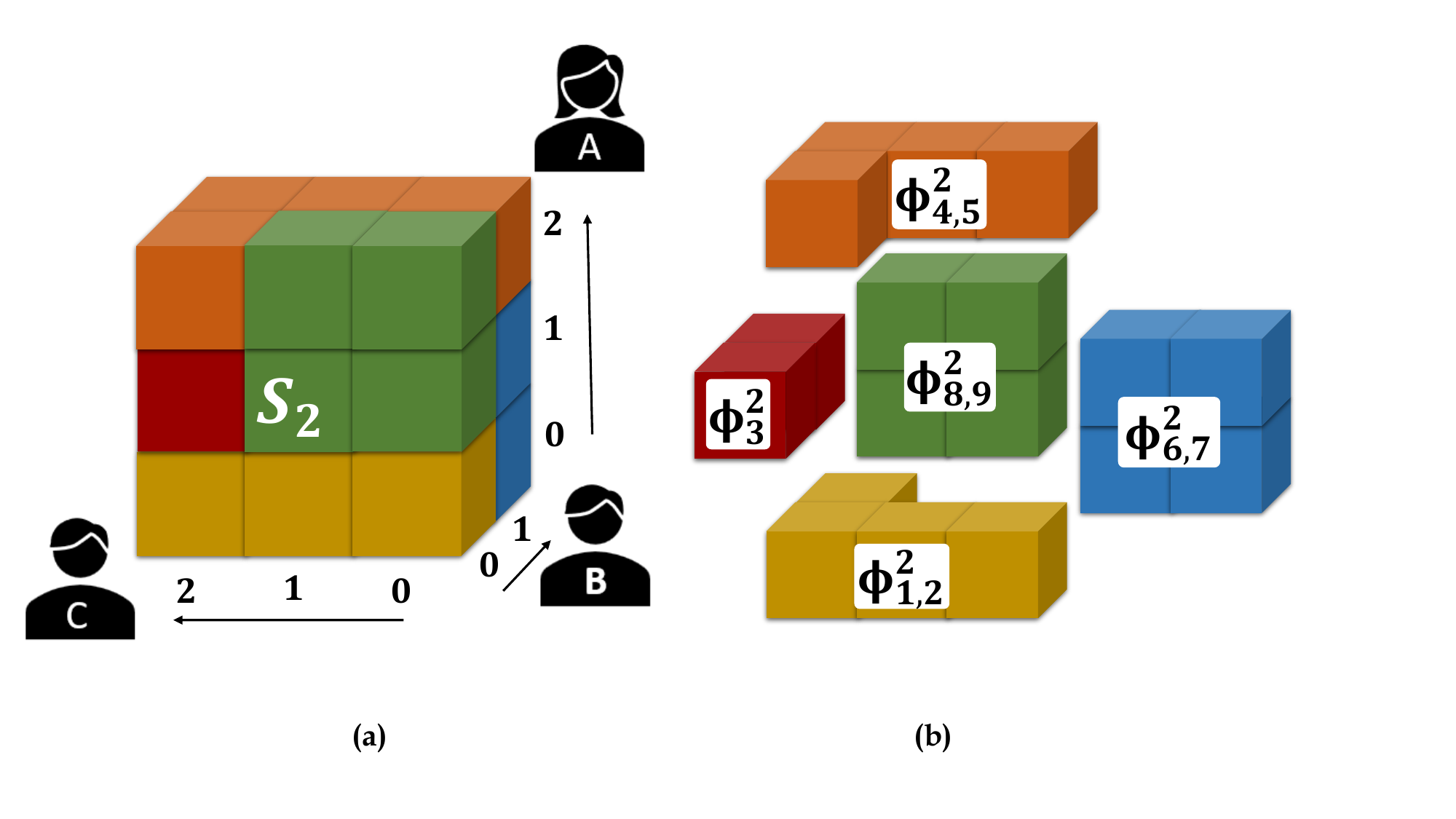}
     \caption{The block structure for $\mathbb{C}^{3 \otimes 2 \otimes 3}$. The states of $S_2$ are placed in the corresponding sub-blocks.}
    \label{fig:S2}
\end{figure}

\begin{equation}
\label{S_2}
    \begin{array}{l}
        \ket{\phi^2_{1}} \equiv \ket{0}\ket{00+01+02-12}\\
        \ket{\phi^2_{2}} \equiv \ket{0}\ket{00-01-02-12}\\
        \ket{\phi^2_{3}} \equiv \ket{1}\ket{02-12}\\
        \ket{\phi^2_{4}} \equiv \ket{2}\ket{10+11+12-02}\\
        \ket{\phi^2_{5}} \equiv \ket{2}\ket{10-11-12-02}\\
        \ket{\phi^2_{6,7}} \equiv \ket{0 \pm 1}\ket{10-11}\\
        \ket{\phi^2_{8,9}} \equiv \ket{1 \pm 2}\ket{00-01}
    \end{array}
\end{equation}
The next two theorems distinguish the TYPE-II class from the TYPE-I class.
\begin{theorem}---
	Nonlocality can never be activated in $S_2$ under LPCC.
\end{theorem}

\begin{proof}	
First, we will show that neither Alice nor Bob can perform a nontrivial OPLM under LPCC when they go first.

Suppose Alice goes first with PVM elements $\{P_l\}$. Since the dimension of the party Alice is three and $\sum P_l=\mathbb{I}$, there must exist a rank-1 projective measurement operator, which can be written as $P_\theta=\ket{\theta}\bra{\theta}$. If $\braket{\theta}{0}\neq 0$, then $(P_\theta\otimes \mathbb{I} \otimes \mathbb{I})\ket{\phi^2_1}=\braket{\theta}{0}\ket{\theta}\ket{00+01+02-12}$. To maintain orthogonality between $\ket{\phi^2_1}$, $\ket{\phi^2_3}$ and $\ket{\phi^2_4}$, $\ket{\theta}$ must be orthogonal to $\ket{1}$ and $\ket{2}$, i.e., $\braket{\theta}{1}=\braket{\theta}{2}=0$. Therefore, $\ket{\theta}=\ket{0}$. However, $\ket{0}\bra{0}$ can not keep $\ket{\phi^2_6}$ and $\ket{\phi^2_7}$ orthogonal. So $\braket{\theta}{0}= 0$. Similarly, $\braket{\theta}{1}=\braket{\theta}{2}=0$. Thus, Alice cannot provide a nontrivial OPLM under LPCC.
	
	If Bob goes first instead of Alice, he can only perform rank-1 projective measurement $P_\theta=\ket{\theta}\bra{\theta}$ as a nontrivial measurement because Bob's subsystem has dimension two. If $\braket{\theta}{0}\neq 0$, then $(\mathbb{I} \otimes P_\theta\otimes  \mathbb{I})\ket{\phi^2_2}=\ket{0}\ket{\theta}(\braket{\theta}{0}\ket{0-1}-\braket{\theta}{0+1}\ket{2})$. To keep orthogonality between $\ket{\phi^2_2}$ and $\ket{\phi^2_6}$, $\ket{\theta}$ must be orthogonal to $\ket{1}$, hence $\ket{\theta}=\ket{0}$. However, $\ket{0}\bra{0}$ can not keep orthogonality between $\ket{\phi^2_1}$ and $\ket{\phi^2_2}$. So $\braket{\theta}{0}= 0$. Similarly one can check $\braket{\theta}{1}=0$. So it is also impossible for Bob to perform a nontrivial OPLM under LPCC.
	
	Now, a comprehensive analysis of all possible OPLMs measurable by Charlie under LPCC is provided in Appendix-D[\ref{theorem_3}]. In a parallel fashion, we show that nontrivial measurements are not possible within Charlie's subsystem.
	
	Therefore all the above studies together indicate that activation of nonlocality is impossible in $S_2$ by LPCC.
\end{proof}

Next, if it can be further proven that at least two parties can perform joint measurements to activate nonlocality in $S_2$, then it can be stated that $S_2$ falls under the more local category than a TYPE-I set, as discussed earlier. The following theorem accomplishes this for $S_2$.

.
\begin{theorem}---
	If Bob and Charlie are allowed to perform joint PVMs, nonlocality can be activated in $S_2$ under the LPCC scenario.
\end{theorem}
The detailed proof is provided in Appendix-E[\ref{theorem_4}]; however, we have also included a pictorial representation[\ref{fig:S2Unfold},\ref{fig:S2Activation}] below.

\begin{figure}[htp]
    \centering
    \includegraphics[scale=.38]{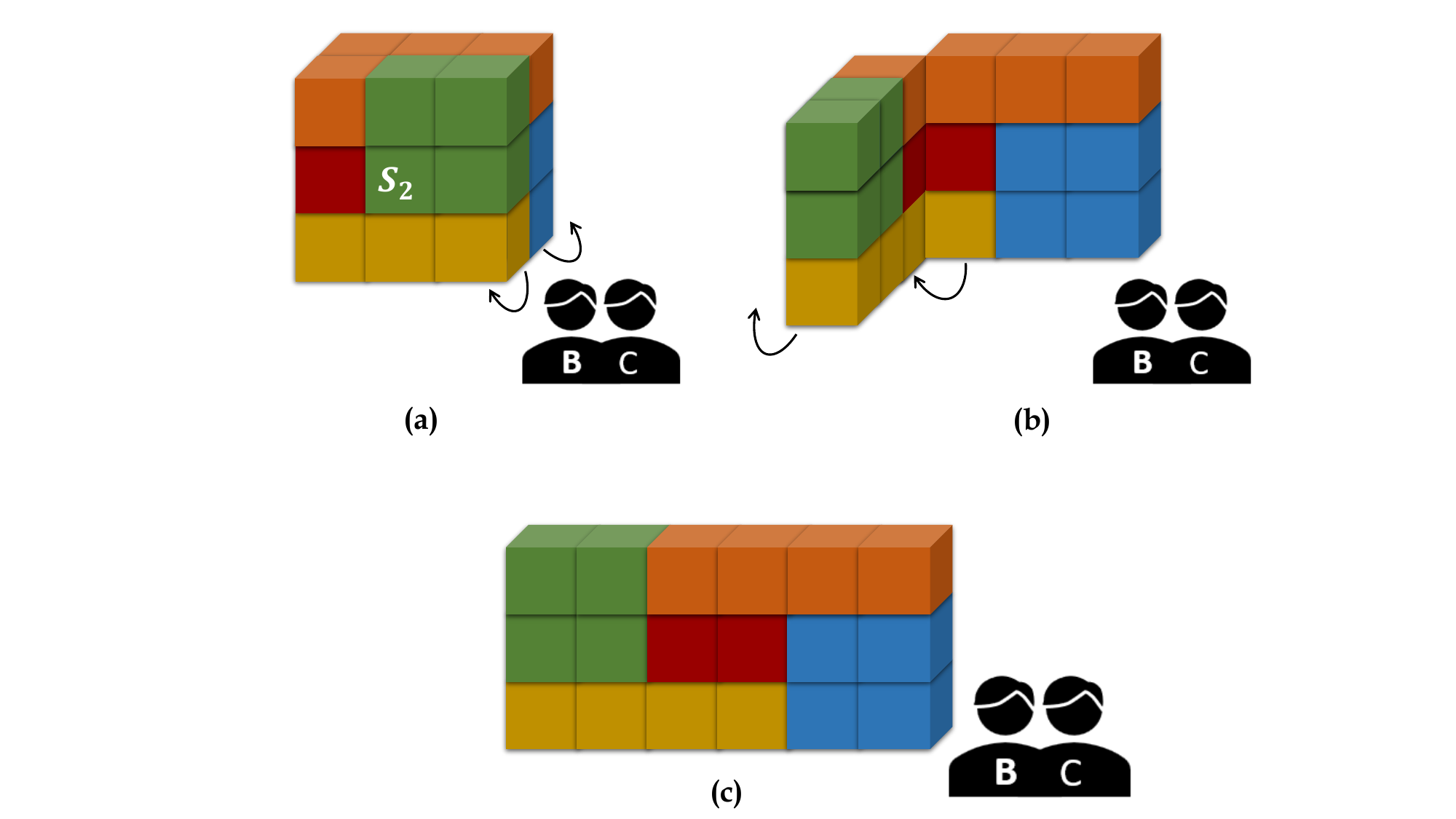}
     \caption{These figures together represent the consideration of the tripartite system as a bipartite system in $A|BC$ partition.}
    \label{fig:S2Unfold}
\end{figure}

\begin{figure}[htp]
    \centering
    \includegraphics[scale=.33]{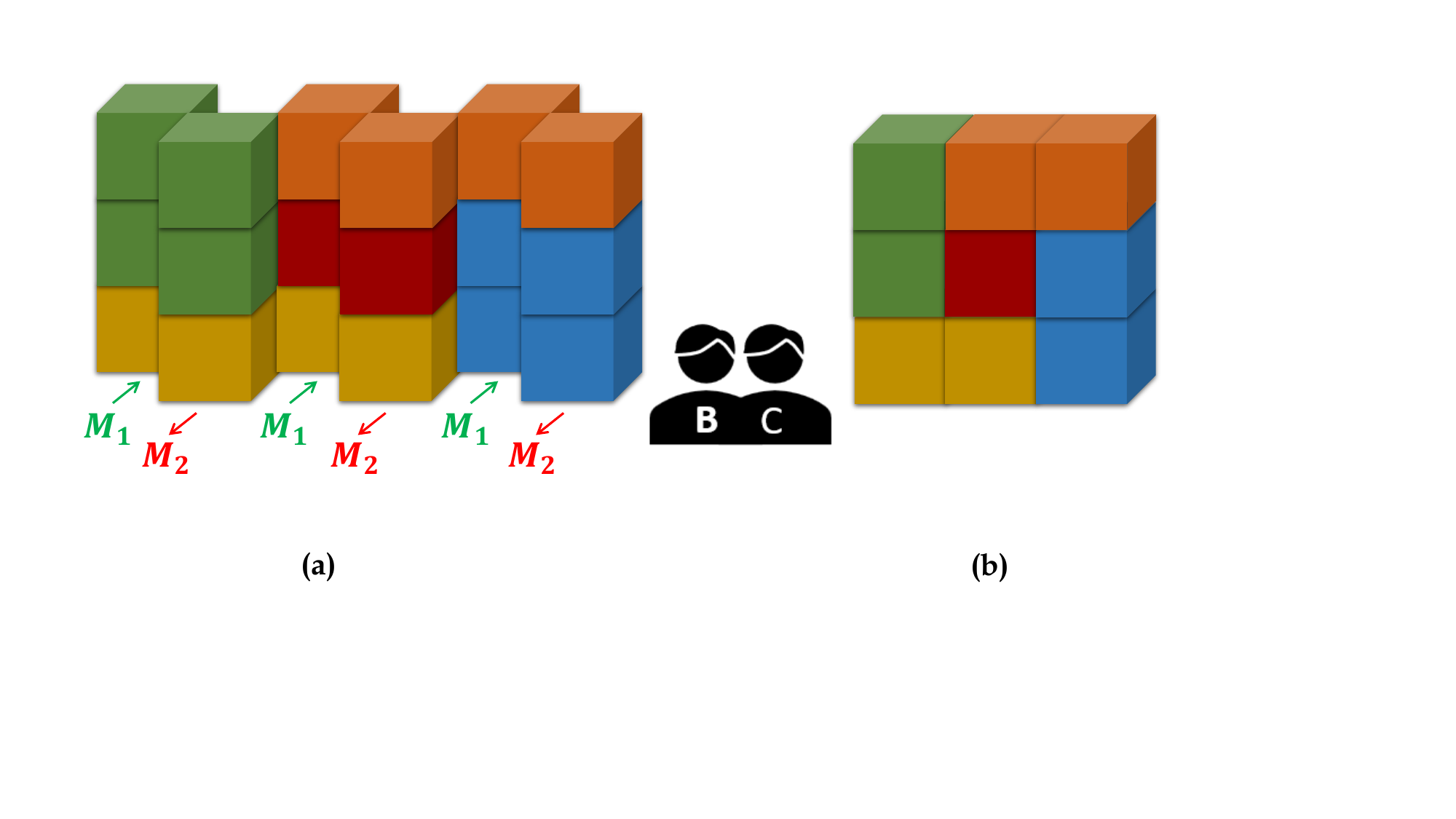}
     \caption{Pictorial representation of hidden nonlocality activation of $S_2$ when Bob and Charlie act from the same lab. $M_1=\ket{00}\bra{00}+\ket{02}\bra{02}+\ket{11}\bra{11}$ and $M_2=\ket{01}\bra{01}+\ket{10}\bra{10}+\ket{12}\bra{12}$ be the measurement setup given by Bob and Charlie together. \textbf{(a)} The `$1$' outcome will isolate the odd columns of the block structure of $S_2$(considering the system in $A|BC$ partition as in FIG.[\ref{fig:S2Unfold}]), while the `$2$' outcome will isolate its even columns. In both cases, we end up on the nonlocal set of nine mutually orthogonal pure product states in $A|BC$ partition. \textbf{(b)} This figure represents the block structure of the aforementioned set which is nothing but a Bennett's set.}
    \label{fig:S2Activation}
\end{figure}


Now we construct two local sets $S'_2$[\ref{S'_2}]  and $S''_2$[\ref{S"_2}] where no single party alone is capable of giving a nontrivial projective measurement to activate nonlocality under LPCC. However, Charlie and Alice can activate nonlocality in $S'_2$, while Alice and Bob can activate nonlocality in $S''_2$ when they act jointly. 

The first set, $S'_2$, contains nine orthogonal pure states in $\mathcal{H}^3_{A}\otimes\mathcal{H}^3_{B}\otimes\mathcal{H}^2_{C}$. Each state is a product state in $B|CA$ partition with the following structure:
\begin{example}
	$S'_2 \subset \mathbb{C}^{3 \otimes 3 \otimes 2}:$
\end{example}

\begin{equation}
\label{S'_2}
	\begin{array}{l}
		\ket{\psi^2_{1}} \equiv \ket{2}_{B}\ket{33+34+35-45}_{CA}\\
		\ket{\psi^2_{2}} \equiv \ket{2}\ket{33-34-35-45}\\
		\ket{\psi^2_{3}} \equiv \ket{3}\ket{35-45}\\
		\ket{\psi^2_{4}} \equiv \ket{4}\ket{43+44+45-35}\\
		\ket{\psi^2_{5}} \equiv \ket{4}\ket{43-44-45-35}\\
		\ket{\psi^2_{6,7}} \equiv \ket{2 \pm 3}\ket{43-44}\\
		\ket{\psi^2_{8,9}} \equiv \ket{3 \pm 4}\ket{33-34}
	\end{array}
\end{equation}
The second set, $S''_2$, also contains same number of orthogonal pure states in $\mathcal{H}^2_{A}\otimes\mathcal{H}^3_{B}\otimes\mathcal{H}^3_{C}$. Each state is a product state in $C|AB$ partition with the following structure:
\begin{example}
	$S''_2 \subset \mathbb{C}^{2 \otimes 3 \otimes 3}:$
\end{example}

\begin{equation}
\label{S"_2}
	\begin{array}{l}
		\ket{\eta^2_{1}} \equiv \ket{5}_{C}\ket{65+66+67-77}_{AB}\\
		\ket{\eta^2_{2}} \equiv \ket{5}\ket{65-66-67-77}\\
		\ket{\eta^2_{3}} \equiv \ket{6}\ket{67-77}\\
		\ket{\eta^2_{4}} \equiv \ket{7}\ket{75+76+77-67}\\
		\ket{\eta^2_{5}} \equiv \ket{7}\ket{75-76-77-67}\\
		\ket{\eta^2_{6,7}} \equiv \ket{5 \pm 6}\ket{75-76}\\
		\ket{\eta^2_{8,9}} \equiv \ket{6 \pm 7}\ket{65-66}
	\end{array}
\end{equation}

We have so far discussed only those sets whose nonlocality is activated only when a particular pair of parties perform a joint measurement. However, we now affirm the existence of a set of quantum states whose nonlocality can be activated upon the joint measurement of any two parties.

\begin{theorem}---
	In the local set of states $S=S_2 \cup S'_2 \cup S''_2$ of the tripartite system $\mathbb{C}^{8 \otimes 8 \otimes 8}$, activation of nonlocality is impossible under LPCC. However, if any pair of parties can measure jointly using PVM, they can activate nonlocality in it.  
\end{theorem}

\begin{proof}
	Suppose Alice goes first with her measurement choices $M=\ket{0}\bra{0}+\ket{1}\bra{1}+\ket{2}\bra{2}$, $M'=\ket{3}\bra{3}+\ket{4}\bra{4}+\ket{5}\bra{5}$ and $M''=\ket{6}\bra{6}+\ket{7}\bra{7}$. Then the respective outcomes separate the sets $S_2$,$S'_2$ and $S''_2$. Since each set is locally distinguishable, $S$ is local under LPCC.
	
	Now it is clear that $S$ is free from local redundancy because some of its subsets ($S_2$,$S'_2$ or $S''_2$) carry the same property. It is also straightforward to show that no single party will be able to activate nonlocality because it is impossible to create in $S_2$, $S'_2$, and $S''_2$. However, Bob and Charlie can jointly activate nonlocality in $S_2$. Similarly, Alice-Charlie and Alice-Bob can perform joint measurements to activate nonlocality in $S'_2$ and $S''_2$ respectively. 
\end{proof}

\begin{example}
	$S_{2,m} \subset \mathbb{C}^{(2m+1) \otimes 2 \otimes (2m+1)}:$
\end{example}

We provide a specific generalization of $S_2$ in $\mathbb{C}^{(2m+1) \otimes 2 \otimes (2m+1)}$ system which have the following form:
\begin{widetext}
	\begin{equation}
    \begin{array}{l}
\ket{\zeta_1}=\vspace{5 pt}\ket{m}\ket{0-1}\ket{2m}\\
\vspace{3 pt}\text{For } i=0,1,\ldots,m-1,\\

\ket{\zeta^1_{i,1}}=\ket{i}\left\{\left(\frac{1-(-1)^{m+i}}{2}\ket{0}+\frac{1+(-1)^{m+i}}{2}\ket{1}\right)\left(\ket{2m-2}+\ket{2m-1}+\ket{2m}\right)-\left(\frac{1+(-1)^{m+i}}{2}\ket{0}+\frac{1-(-1)^{m+i}}{2}\ket{1}\right)\ket{2m}\right\}\\
    \ket{\zeta^1_{i,2}}=\ket{i}\left\{\left(\frac{1-(-1)^{m+i}}{2}\ket{0}+\frac{1+(-1)^{m+i}}{2}\ket{1}\right)\left(\ket{2m-2}-\ket{2m-1}-\ket{2m}\right)-\left(\frac{1+(-1)^{m+i}}{2}\ket{0}+\frac{1-(-1)^{m+i}}{2}\ket{1}\right)\ket{2m}\right\}\\
    \ket{\zeta^1_{i,3}}=\ket{2m-i}\left\{\left(\frac{1+(-1)^{m+i}}{2}\ket{0}+\frac{1-(-1)^{m+i}}{2}\ket{1}\right)\left(\ket{2m-2}+\ket{2m-1}+\ket{2m}\right)-\left(\frac{1-(-1)^{m+i}}{2}\ket{0}+\frac{1+(-1)^{m+i}}{2}\ket{1}\right)\ket{2m}\right\}\\
\vspace{5 pt}\ket{\zeta^1_{i,4}}=\ket{2m-i}\left\{\left(\frac{1+(-1)^{m+i}}{2}\ket{0}+\frac{1-(-1)^{m+i}}{2}\ket{1}\right)\left(\ket{2m-2}-\ket{2m-1}-\ket{2m}\right)-\left(\frac{1-(-1)^{m+i}}{2}\ket{0}+\frac{1+(-1)^{m+i}}{2}\ket{1}\right)\ket{2m}\right\}\\

\text{For }k=0,1,\ldots,2m-2i-1,\\
\vspace{5 pt}\ket{\zeta^1_{i,4+k}}=\left(\ket{i+k}\pm \ket{i+k+1}\right)\left(\frac{1-(-1)^{k}}{2}\ket{0}+\frac{1+(-1)^{k}}{2}\ket{1}\right)\left(\ket{2i}-\ket{2i+1}\right)\\

    \vspace{3 pt}\text{For }m\geqslant 2 \text{ and } k_1=0,1,\cdots,\left[\frac{m-i}{2}\right]-1,\\
    \ket{\zeta^2_{i,k_1}}=\ket{i}\ket{0}\left(\ket{2i+4k_1}+\ket{2i+4k_1+1}+\ket{2i+4k_1+2}-\ket{2i+4k_1+3}\right)\\
    \ket{\zeta^3_{i,k_1}}=\ket{i}\ket{0}\left(\ket{2i+4k_1}-\ket{2i+4k_1+1}-\ket{2i+4k_1+2}-\ket{2i+4k_1+3}\right)\\
    \ket{\zeta^4_{i,k_1}}=\ket{2m-i}\ket{1}\left(\ket{2i+4k_1}+\ket{2i+4k_1+1}+\ket{2i+4k_1+2}-\ket{2i+4k_1+3}\right)\\
    \vspace{5 pt}\ket{\zeta^5_{i,k_1}}=\ket{2m-i}\ket{1}\left(\ket{2i+4k_1}-\ket{2i+4k_1+1}-\ket{2i+4k_1+2}-\ket{2i+4k_1+3}\right)\\
    
    \vspace{3 pt}\text{For }m\geqslant 3 \text{ and } k_2=0,1,\ldots,\left[\frac{m-i-1}{2}\right]-1,\\
    \ket{\zeta^6_{i,k_2}}=\ket{i}\ket{1}\left(\ket{2i+4k_2+2}+\ket{2i+4k_2+3}+\ket{2i+4k_2+4}-\ket{2i+4k_2+5}\right)\\
    \ket{\zeta^7_{i,k_2}}=\ket{i}\ket{1}\left(\ket{2i+4k_2+2}-\ket{2i+4k_2+3}-\ket{2i+4k_2+4}-\ket{2i+4k_2+5}\right)\\
    \ket{\zeta^8_{i,k_2}}=\ket{2m-i}\ket{0}\left(\ket{2i+4k_2+2}+\ket{2i+4k_2+3}+\ket{2i+4k_2+4}-\ket{2i+4k_2+5}\right)\\
    \ket{\zeta^9_{i,k_2}}=\ket{2m-i}\ket{0}\left(\ket{2i+4k_2+2}-\ket{2i+4k_2+3}-\ket{2i+4k_2+4}-\ket{2i+4k_2+5}\right)\\
    \end{array}
\end{equation}

The above set is structured in FIG.[\ref{fig:Generalization}] in the $A|BC$ partition. After a closer look, we can check that the aforementioned set is perfectly aligned with the set $S_2$ if we set $m=1$, i.e., $S_{2,1} \equiv S_2$.

\begin{figure}[htp]
    \centering
    \includegraphics[scale=.64]{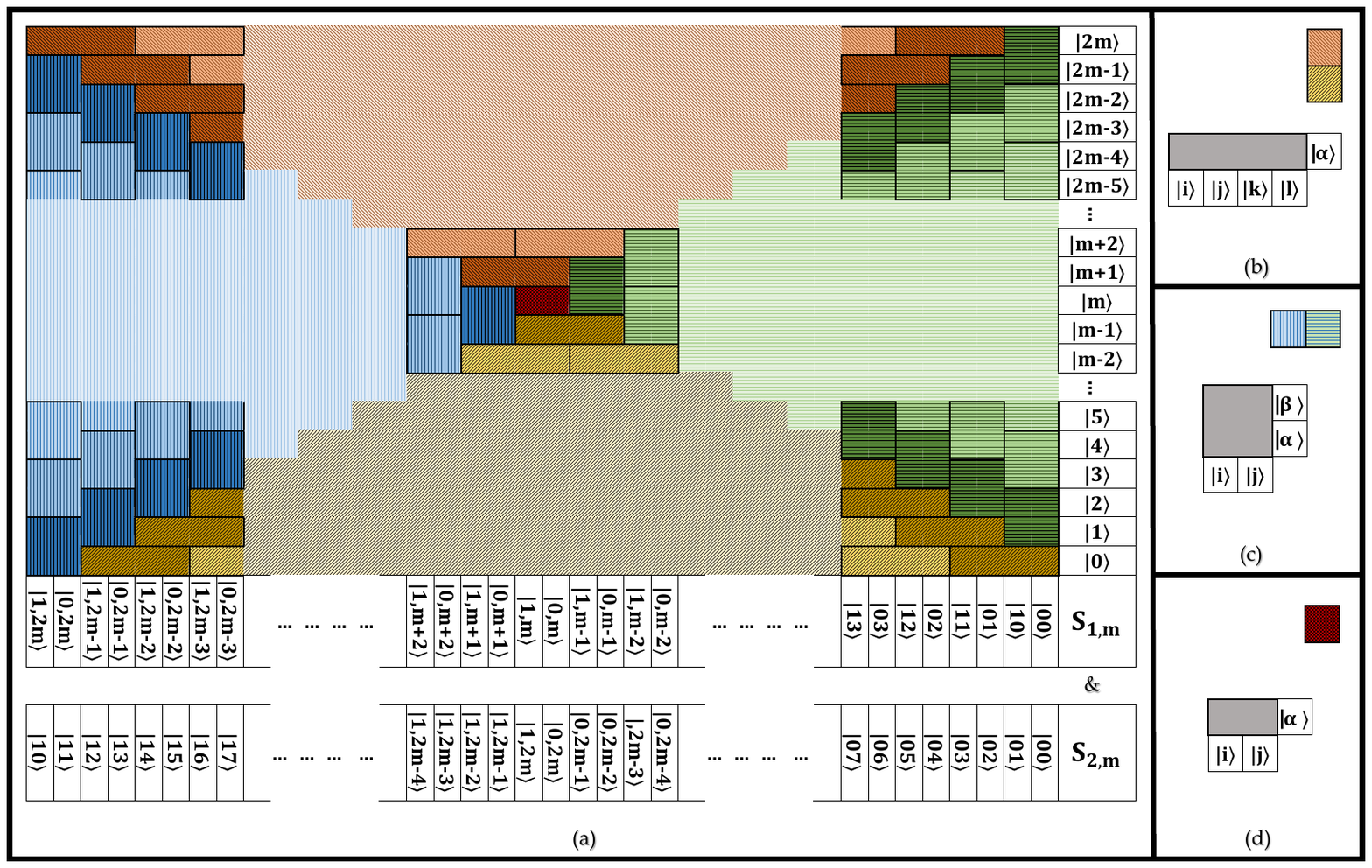}
     \caption{(a) The figure showcases a fascinating structure of two distinct sets of orthogonal quantum states of a tripartite system, denoted as $S_{1,m}$ and $S_{2,m}$. Although both sets share a similar underlying structure, their physical properties diverge depending on the chosen basis in the $A|BC$ partition. As we can see, the whole structure is made of three types of blocks. Here, we represent (b) two quantum states $\ket{\alpha}_A\ket{(i-l)\pm (j+k)}_{BC}$ by the rectangular blocks made of four small squares; (c) two quantum states $\ket{\alpha \pm \beta}_A\ket{i-j}_{BC}$ by the square blocks made of four small squares; and (d) one quantum states $\ket{\alpha}_A\ket{i-j}_{BC}$ by the rectangular block made of two small squares.}
    \label{fig:Generalization}
\end{figure}

\end{widetext}
\section{general definition of activation of nonlocality in multipartite scenario}
In this section, we present a general definition of quantum hidden nonlocality based on the local activation for the orthogonal quantum states in $d_1 \otimes d_2 \otimes \cdots \otimes d_n, n \geqslant 3$. This definition explores how hidden nonlocality can be revealed through OPLMs within such multipartite systems.

\emph{Definition 5}--- \emph{$m$-partition} of any $n$-partite quantum system is the consideration of the system as an m-partite one, where more than one party might be considered as a single party.

\emph{Definition 6}--- In $d_1 \otimes d_2 \otimes \cdots \otimes d_n, n \geqslant 3$, a locally distinguishable set of quantum states is called \emph{$m$-activable} if it is possible to activate hidden nonlocality in at least one $m$-partition such that the postmeasurement states (after activation) become locally irreducible within that partition. 

\emph{Definition 7}--- In $d_1 \otimes d_2 \otimes \cdots \otimes d_n, n \geqslant 3$, a locally distinguishable set of quantum states is called \emph{strong-$m$-activable} if it is possible to activate hidden nonlocality in at least one $m$-partition such that the postmeasurement states (after activation) are not only locally irreducible within that partition, but also locally irreducible in at least one $(m-1)$-partition.

\emph{Proposition 2}--- In $d_1 \otimes d_2 \otimes \cdots \otimes d_n, n \geqslant 3$, a strong-$n$-activable set is always $(n-1)$-activable but the converse is not necessarily true.

Suppose $A_i$'s are the local observer associated with the local subsystems of dimension $d_i$. Given that the aforementioned set (say $\mathcal{S}$) is strong-$n$-activable and therefore the postmeasurement states (after activation) are locally irreducible in at least one $(n-1)$-partition. Without loss of generality, we can consider the corresponding partition to be $P_{n-1}=A_1A_2|A_3|A_4|\cdots| A_n$. We  know that any OPLM in the $n$-partition $P_{n}=A_1|A_2|A_3|A_4|\cdots| A_n$ is also an OPLM in $P_{n-1}$. Therefore, the OPLMs responsible for activating hidden nonlocality in $P_{n}$ can also activate nonlocality in $P_{n-1}$. Consequently, in the $n$-partition $P_{n}$, the set  $\mathcal{S}$ is locally activable, and the postmeasurement states (after activation) are locally irreducible. Therefore, by definition $\mathcal{S}$ is $(n-1)$-activable.

Furthermore, we demonstrate a class of sets $\{\mathcal{U}_j^n\}_{j=2}^n$ within an $n$-partite system, where $\mathcal{U}_j^n$ is $j$-activable but not $m$-activable for any $n \geqslant m > j$. Accordingly, the degree of nonlocal operations required to activate the hidden nonlocality in $\mathcal{U}_j^n$ is greater than that required for $\mathcal{U}_{m>j}^n$. In contrast, $\mathcal{U}_j^n$ can be considered to have a lower degree of nonlocality compared to $\mathcal{U}_{m>j}^n$. Thus, the relationship between these sets can be described as:  
$$ \mathcal{H}_1^n<\mathcal{H}_2^n<\cdots<\mathcal{H}_{j-1}^n<\mathcal{H}_j^n<\mathcal{H}_{j+1}^n<\cdots<\mathcal{H}_n^n
$$

where, for $j=2, \ldots, n, \mathcal{H}_j^n$ denotes the degree of nonlocality of  $\mathcal{U}_j^n$. Note that, $ \mathcal{H}_1^n$ represents the degree of nonlocality associated with $ \mathcal{U}_1^n$, whose hidden nonlocality cannot be activated in any partition. 

Now, considering the scenario from the perspective of locality, the relationship reverses to: 
$$
\mathcal{L}_1>\mathcal{L}_2>\cdots>\mathcal{L}_{j-1}>\mathcal{L}_{j}>\mathcal{L}_{j+1}>\cdots>\mathcal{L}_{n}
$$
where, for $j=1, \ldots, n, \mathcal{L}_j$ denotes the degree of locality of $\mathcal{U}_j^n$. This indicates that $\mathcal{U}_j^n$ is more local than $\mathcal{U}_{m>j}^n$. Notably, $\mathcal{U}_n^n$ represents the class of locally activable sets in $n$-partite system, whereas $\mathcal{U}_1^n$ represents the class of strong local sets.

\section{discussion}
 In this paper, we construct two local sets, $S_1$ and $S_2$ in a $3\otimes2\otimes3$ system. For the set $S_1$, a single party can activate nonlocality through OPLM under LPCC, whereas for $S_2$, this is not possible. However, if two parties jointly adopt projective measurements, they can successfully activate nonlocality in $S_2$. In other words, $S_2$ demands a higher degree of nonlocal operations to reveal its hidden nonlocality compared to $S_1$. Thus, it can be considered to have a lower degree of nonlocality and, conversely, a higher degree of locality than $S_1$. This introduces a stronger manifestation of locality in multipartite systems through the notion of activating hidden nonlocality, asserting that $S_2$ is more local than $S_1$. Throughout the study, we consistently emphasize the nontriviality in constructing such activable sets. Our focus remains steadfast on addressing the genuineness of the hidden nonlocality activation by thoroughly examining the sets to ensure they have no redundancy. Many interesting problems have risen during the study. One such issue is finding the existence of a LOCC distinguishable class, whose nonlocality cannot be activated through OPLMs but requires joint measurements by two or more parties. During our work, we observed that local activation does not imply joint activation (when more than one parties measure jointly). The converse statement is also not true. Therefore, the complexity of the structure of a set bearing genuine hidden nonlocality in a system involving three or more parties remains inadequately understood. Throughout our study, we also observe that PVMs are sufficient to activate nonlocality in the locally distinguishable sets presented in \cite{BandyopadhyayHalderActivation2021, Li2022, GhoshStrongActivation2022}. However, we have not yet studied such sets where nonlocality cannot be activated through PVMs, but instead require another class of local measurements. Identifying and understanding such sets could provide deeper insights into the nature of hidden nonlocality and the types of measurements needed to reveal it.

 	\section*{ACKNOWLEDGEMENTS}
 The author S. Bera acknowledges the support from CSIR, India. The authors A. Bhunia and I. Biswas acknowledge the support from UGC, India. The authors I. Chattopadhyay and D. Sarkar acknowledge the work as part of initiatives by DST India.

\section{Appendix}
\subsection{Distinguishability of set $S_1$.}
\label{distinguishability_S_1}
First, we give a set of Charlie's projective measurement operators:

\begin{equation*}
	\begin{array}{c}
		P_1=\ket{0}\bra{0}+\ket{1}\bra{1}\\
		P_2=\ket{2}\bra{2}
	\end{array}
\end{equation*}

If the outcome is 1, then postmeasurement states would be

\begin{equation}
	\begin{array}{c}
		\ket{0}\ket{00+01+10-11}\\
		\ket{0}\ket{00-01-10-11}\\
		\ket{1}\ket{0-1}\ket{1}\\
		\ket{2}\ket{0 \pm 1}\ket{1}\\
		\ket{1 \pm 2}\ket{0-1}\ket{0}
	\end{array}
\end{equation}

Now if Bob measures in $\ket{0-1}\bra{0-1}$ and $\ket{0+1}\bra{0+1}$, the respective postmeasurement states would be

\begin{equation}
	\label{S1DC1B-}
	\begin{array}{c}
		\ket{0}\ket{0-1}\ket{1}\\
		\ket{0}\ket{0-1}\ket{0}\\
		\ket{1}\ket{0-1}\ket{1}\\
		\ket{2}\ket{0 - 1}\ket{1}\\
		\ket{1 \pm 2}\ket{0-1}\ket{0}
	\end{array}
\end{equation}
and
\begin{equation}
	\label{S1DC1B+}
	\begin{array}{c}
		\ket{0}\ket{0+1}\ket{0}\\
		\ket{0}\ket{0+1}\ket{1}\\
		\ket{2}\ket{0+1}\ket{1}
	\end{array}
\end{equation}

The states in (\ref{S1DC1B-}) and (\ref{S1DC1B+}) are the orthogonal product states in $\mathbb{C}^{3 \otimes 1 \otimes 2}$ and hence by lemma-1 these can be distinguished by LPCC.

While the outcome is 2, then postmeasurement states would be

\begin{equation}
	\begin{array}{c}
		\ket{2}\ket{0 \pm 1}\ket{2}\\
		\ket{0 \pm 1}\ket{0-1}\ket{2}
	\end{array}
\end{equation}

Clearly, these are the orthogonal product states in $\mathbb{C}^{3 \otimes 2 \otimes 1}$, and hence by lemma-1 these can be distinguished by LPCC.
\subsection{Proof of theorem 2.}
\label{theorem_2}
Here, Bob goes first and measures in $M_0=\ket{0}\bra{0}$ and $M_1=\ket{1}\bra{1}$.

If the outcome is 0, the postmeasurement states would be

\begin{equation}
	\begin{array}{c}
		\ket{0}\ket{00 \pm 01}\\
		\ket{0\pm1}\ket{02}\\
		\ket{1\pm2}\ket{00}\\
		\ket{2}\ket{01 \pm 02}\\
		\ket{1}\ket{01}
	\end{array}
\end{equation}

Now we can write these states in terms of $\{\ket{0},\ket{1},\ket{2}\}$ basis instead of  $\{\ket{00},\ket{01},\ket{02}\}$ as

\begin{equation}
	\label{BennettPB1999}
	\begin{array}{c}
		\ket{0}\ket{0 \pm 1}\\
		\ket{0\pm1}\ket{2}\\
		\ket{1\pm2}\ket{0}\\
		\ket{2}\ket{1 \pm 2}\\
		\ket{1}\ket{1}
	\end{array}
\end{equation}

This is a nonlocal set of nine pure orthogonal product states of $3 \otimes 3$ system proposed by Bennett et al. in \cite{BennettPB1999}.

If the outcome is 1, the postmeasurement states would be

\begin{equation}
	\begin{array}{c}
		\ket{0}\ket{10 \pm 11}\\
		\ket{0\pm1}\ket{12}\\
		\ket{1\pm2}\ket{10}\\
		\ket{2}\ket{11 \pm 12}\\
		\ket{1}\ket{11}
	\end{array}
\end{equation}

Now we can rewrite the above states in terms of $\{\ket{0},\ket{1},\ket{2}\}$ basis instead of $\{\ket{10},\ket{11},\ket{12}\}$ as in (\ref{BennettPB1999}) and hence locally indistinguishable.


\subsection{Distinguishability of set $S_2$.}
\label{distinguishability_S_2}
First, we give a set of Charlie's projective measurement operators:

\begin{equation*}
	\begin{array}{c}
		P_1=\ket{0}\bra{0}+\ket{1}\bra{1}\\
		P_2=\ket{2}\bra{2}
	\end{array}
\end{equation*}

If the outcome is 1, then postmeasurement states would be

\begin{equation}
	\begin{array}{c}
		\ket{0}\ket{0}\ket{0 \pm 1}\\
		\ket{0 \pm 1}\ket{1}\ket{0-1}\\
		\ket{1 \pm 2}\ket{0}\ket{0-1}\\
		\ket{2}\ket{1}\ket{0 \pm 1}
	\end{array}
\end{equation}

Now if Bob measures in $\ket{0}\bra{0}$ and $\ket{1}\bra{1}$, the respective postmeasurement states would be

\begin{equation}   
	\label{S2DC1B0}
	\begin{array}{c}
		\ket{0}\ket{0}\ket{0 \pm 1}\\
		\ket{1 \pm 2}\ket{0}\ket{0-1}
	\end{array}
\end{equation}

\begin{center}
	and
\end{center}

\begin{equation}
	\label{S2DC1B1}
	\begin{array}{c}
		\ket{0 \pm 1}\ket{1}\ket{0-1}\\
		\ket{2}\ket{1}\ket{0 \pm 1}
	\end{array}
\end{equation}

The states in (\ref{S2DC1B0}) and (\ref{S2DC1B1}) are the orthogonal product states in $\mathbb{C}^{3 \otimes 2 \otimes 1}$ and hence by lemma-1 these can be distinguished by LPCC.

While the outcome is 2, then postmeasurement states would be

\begin{equation}
	\begin{array}{c}
		\ket{0}\ket{0 \pm 1}\ket{2}\\
		\ket{2}\ket{0 \pm 1}\ket{2}\\
		\ket{1}\ket{0 - 1}\ket{2}
	\end{array}
\end{equation}

Obviously, these are orthogonal product states in $\mathbb{C}^{3 \otimes 2 \otimes 1}$, and hence by lemma-1 these can be distinguished by LPCC.

\subsection{Proof of theorem 3.}
\label{theorem_3}
As with Alice, there must be a corresponding rank-1 projective measurement that can be written as $P_1=\ket{\theta}\bra{\theta}$ when Charlie goes first. Then the postmeasurement states $(\mathbb{I} \otimes  \mathbb{I} \otimes \ket{\theta}\bra{\theta})\ket{\phi^2_i}$ can be written as

\begin{equation}
	\begin{array}{c}
		\ket{0}(\braket{\theta}{0\pm (1+2)}\ket{0}-\braket{\theta}{2}\ket{1})\ket{\theta}\\
		\braket{\theta}{2}\ket{1}\ket{0-1})\ket{\theta}\\
		\ket{2}(\braket{\theta}{0\pm (1+2)}\ket{1}-\braket{\theta}{2}\ket{0})\ket{\theta}\\
		\braket{\theta}{0-1}\ket{0 \pm 1}\ket{1}\ket{\theta}\\
		\braket{\theta}{0-1}\ket{1 \pm 2}\ket{0}\ket{\theta}
	\end{array}
\end{equation}

Clearly, at least one of $\ket{\phi_3}$ and $\ket{\phi_6}$ will be omitted, while Charlie's rank-1 projective measurement $P_1$ depends on the conditions $\braket{\theta}{0-1}=0$ and $\braket{\theta}{2}=0$, respectively.

\textbf{Case-1($\braket{\theta}{0-1} \neq 0$)}

Therefore, $\braket{\theta}{2} = 0$ and the postmeasurement states will be 

\begin{equation}
	\begin{array}{c}
		\braket{\theta}{0\pm 1}\ket{0}\ket{0}\ket{\theta}\\
		\braket{\theta}{0\pm 1}\ket{2}\ket{1}\ket{\theta}\\
		\braket{\theta}{0-1}\ket{0 \pm 1}\ket{1}\ket{\theta}\\
		\braket{\theta}{0-1}\ket{1 \pm 2}\ket{0}\ket{\theta}
	\end{array}
\end{equation}

As $\braket{\theta}{0-1} \neq 0$, $\ket{\phi^2_1}$ should be eliminated by Charlie's measurement discussed earlier and hence $\braket{\theta}{0+1}=0$.

Now we have $\braket{\theta}{2}=0$ and $\braket{\theta}{0+1}=0$, which implies $\ket{\theta}=\ket{0-1}$. Therefore, the postmeasurement states will be

\begin{equation}
	\label{S2EC-}
	\begin{array}{c}
		\ket{0}\ket{0}\ket{0-1}\\
		\ket{0 \pm 1}\ket{1}\ket{0-1}\\
		\ket{1 \pm 2}\ket{0}\ket{0-1}\\
		\ket{2}\ket{1}\ket{0-1}
	\end{array}
\end{equation}

The states in (\ref{S2EC-}) are orthogonal product states in $\mathbb{C}^{3 \otimes 2 \otimes 1}$ and hence these are distinguishable by LPCC by lemma-1 and activation of nonlocality by corollary-1 is therefore unlikely.

\textbf{Case-2($\braket{\theta}{0-1} = 0$)}

Therefore, the postmeasurement states will be

\begin{equation}
	\begin{array}{c}
		\ket{0}(\braket{\theta}{0+1}\ket{0}+\braket{\theta}{2}\ket{0-1})\ket{\theta}\\
		\braket{\theta}{2}\ket{0}\ket{0+1})\ket{\theta}\\
		\braket{\theta}{2}\ket{1}\ket{0-1}\ket{\theta}\\
		\braket{\theta}{2}\ket{2}\ket{0+1})\ket{\theta}\\
		\ket{2}(\braket{\theta}{0+1}\ket{1}-\braket{\theta}{2}\ket{0-1})\ket{\theta}
	\end{array}
\end{equation}

Here we can easily notice that $\ket{\phi_0}$ and $\ket{\phi_1}$ cannot coexist and so either $\braket{\theta}{0+1}=0$ or $\braket{\theta}{2}=0$.

\textbf{Subcase-2A($\braket{\theta}{0-1} = 0$ and $\braket{\theta}{2} \neq 0$)}

Therefore, $\braket{\theta}{0+1}=0$ and hence, $\ket{\theta}=\ket{2}$ and the postmeasurement states will be 

\begin{equation}
	\label{S2EC2}
	\begin{array}{c}
		\ket{0}\ket{0\pm1}\ket{2}\\
		\ket{1}\ket{0-1}\ket{2}\\
		\ket{2}\ket{0\pm1}\ket{2}
	\end{array}
\end{equation}

The states in (\ref{S2EC2}) are orthogonal product states in $\mathbb{C}^{3 \otimes 2 \otimes 1}$ and hence by lemma-1 these can be distinguished by LPCC and by corollary-1 it is impossible to further enable nonlocality. 

\textbf{Subcase-2B($\braket{\theta}{0-1} = 0$ and $\braket{\theta}{2} = 0$)}

Therefore, $\ket{\theta}=\ket{0+1}$ and postmeasurement states will be 

\begin{equation}
	\label{S2EC+}
	\begin{array}{c}
		\ket{0}\ket{0}\ket{0+1}\\
		\ket{2}\ket{1}\ket{0+1}
	\end{array}
\end{equation}

These are two orthogonal states and hence these can be distinguished by LPCC, also from here it is impossible to activate nonlocality.

So in conclusion we say that if Charlie goes first, then he can give three rank-1 projective measurements $P_1=\ket{0-1}\bra{0-1}$, $P_2=\ket{0+1}\bra{0+1}$ and $P_3=\ket{2}\bra{2}$ or a rank-1 and a rank-2 projective measurements. 

In the first case, we have already discussed that postmeasurement states are local, and since those are product states in $\mathbb{C}^{3 \otimes 2 \otimes 1}$, we cannot enable nonlocality.

Now we focus on the second case, where a rank-1 projective measurement $M_1$ can be derived from $\ket{0-1}\bra{0-1}$, $\ket{0+1}\bra{0+1}$ and $\ket{2}\bra{2}$ and the corresponding Rank-2 projective measurement will be $M_2=\mathbb{I}-M_1$. Since we have already discussed the postmeasurement states when the outcome is 1, we will now discuss in detail when Charlie measures in the rank-2 projective measurement operator $M_2$.

\textbf{Case-1(when $M_1=\ket{2}\bra{2}$)}

Therefore,  $M_2=\ket{0}\bra{0}+\ket{1}\bra{1}$ and the corresponding postmeasurement states would be 

\begin{equation}
	\begin{array}{c}
		\label{S2EC01}
		\ket{0}\ket{0}\ket{0 \pm 1}\\
		\ket{2}\ket{1}\ket{0 \pm 1}\\
		\ket{0 \pm 1}\ket{1}\ket{0 - 1}\\
		\ket{1 \pm 2}\ket{0}\ket{0 - 1}
	\end{array}
\end{equation}

These can be distinguished by LPCC.[Note that if Charlie measures in $M_1=\ket{0-1}\bra{0-1}$ and $M_2=\ket{0+1}\bra{0+1}$, the postmeasurement states would be similar to (\ref{S2EC-}) and (\ref{S2EC+}) respectively]. Now we have to check whether the activation of nonlocality is possible for the set of orthogonal states in (\ref{S2EC01}). Since the states in (\ref{S2EC01}) are orthogonal product states in $\mathbb{C}^{3 \otimes 2 \otimes 2}$, Alice can only activate nonlocality by corollary-2. Since Alice has dimension 3, there must be a rank-1 projective measurement operator. It is easy to check that there is no such operator that holds the orthogonality of the states in (\ref{S2EC01}). Hence activation of nonlocality is not possible for the set (\ref{S2EC01}).

\textbf{Case-2(when $M_1=\ket{0-1}\bra{0-1}$)}

Therefore,  $M_2=\ket{0+1}\bra{0+1}+\ket{2}\bra{2}$ and the corresponding postmeasurement states would be

\begin{equation}
	\begin{array}{c}
		\label{S2EC+2}
		\ket{0}(\ket{0}\ket{0+1}+\ket{0-1}\ket{2})\\
		\ket{0}\ket{0+1}\ket{2}\\
		\ket{1}\ket{0-1}\ket{2}\\
		\ket{2}\ket{0+1}\ket{2}\\
		\ket{2}(\ket{0-1}\ket{2}-\ket{1}\ket{0+1})
	\end{array}
\end{equation}

These are orthogonal pure states of $\mathbb{C}^{3 \otimes 2 \otimes 2}$ system, each of them is a product state in $A|BC$ partition and thus can be distinguished under LPCC. (Note that if Charlie measures in $M_1=\ket{2}\bra{2}$ and $M_2=\ket{0+1}\bra{0+1}$, the postmeasurement states would be similar to (\ref{S2EC2}) and (\ref{S2EC+}) respectively). Now to enable nonlocality in (\ref{S2EC+2}), by theorem-1 we know that Only Alice can give a non-trivial OPLM. Then Alice's measurement choices must be $\ket{0}\bra{0}$,$\ket{1}\bra{1}$ and $\ket{2}\bra{2}$ and the corresponding postmeasurement states are respectively

\begin{equation}
	\begin{array}{c}
		\ket{0}(\ket{0}\ket{0+1}+\ket{0-1}\ket{2})\\
		\ket{0}\ket{0+1}\ket{2}\\
	\end{array}
\end{equation}

\begin{center}
	and
\end{center}
\begin{equation}
	\begin{array}{c}
		\ket{1}\ket{0-1}\ket{2}
	\end{array}
\end{equation}
\begin{center}
	and
\end{center}

\begin{equation}
	\begin{array}{c}
		\ket{2}\ket{0+1}\ket{2}\\
		\ket{2}(\ket{0-1}\ket{2}-\ket{1}\ket{0+1})
	\end{array}
\end{equation}

But these can always be distinguished by LPCC. So in this case too, measurement choices will not help active nonlocality.

\textbf{Case-3(when $M_1=\ket{0+1}\bra{0+1}$)}

Therefore,  $M_2=\ket{0-1}\bra{0-1}+\ket{2}\bra{2}$ and the corresponding postmeasurement states would be

\begin{equation}
	\label{S2EC-2}
	\begin{array}{c}
		\ket{0}\ket{0-1}\ket{2}\\
		\ket{0}(\ket{0}\ket{0-1}-\ket{0+1}\ket{2})\\
		\ket{1}\ket{0-1}\ket{2}\\
		\ket{2}(\ket{0+1}\ket{2}-\ket{1}\ket{0-1})\\
		\ket{2}\ket{0-1}\ket{2}\\
		\ket{0\pm1}\ket{1}\ket{0-1}\\
		\ket{1\pm2}\ket{0}\ket{0-1}
	\end{array}
\end{equation}

These are orthogonal pure states of $\mathbb{C}^{3 \otimes 2 \otimes 2}$ system, each of them is a product state in $A|BC$ partition and thus can be distinguished under LPCC. (Note that if Charlie measures in $M_1=\ket{0-1}\bra{0-1}$ and $M_2=\ket{2}\bra{2}$, the postmeasurement states would be similar to (\ref{S2EC-}) and (\ref{S2EC2}) respectively). Now to activate nonlocality in (\ref{S2EC-2}), by theorem-1 we know that Only Alice can give a non-trivial OPLM. But here it is impossible for Alice to give a non-trivial projective measurement while keeping all states orthogonal, so it is impossible to enable nonlocality for the set (\ref{S2EC-2}). 

All the above studies together indicate that nonlocality activation under LPCC is impossible for the $S_2$ set.

\subsection{Proof of theorem 4.}
\label{theorem_4}
Here, Bob and Charlie will give the joint projective measurement operators 
\begin{equation}
	\begin{array}{c}
		M_1=\ket{00}\bra{00}+\ket{02}\bra{02}+\ket{11}\bra{11}\\
		M_2=\ket{01}\bra{01}+\ket{10}\bra{10}+\ket{12}\bra{12}
	\end{array}
\end{equation}

If the outcome is 1, the postmeasurement states would be

\begin{equation}
	\label{S2EB00|02|11}
	\begin{array}{c}
		\ket{0}\ket{00 \pm 02}\\
		\ket{1}\ket{02}\\
		\ket{2}\ket{02 \pm 11}\\
		\ket{0\pm1}\ket{11}\\
		\ket{1\pm2}\ket{00}
	\end{array}
\end{equation}

Now we can rewrite the above set of states in $\{\ket{0},\ket{1},\ket{2}\}$ basis in replacing of $\{\ket{00},\ket{02},\ket{11}\}$ as

\begin{equation}
	\begin{array}{c}
		\ket{0}\ket{0 \pm 1}\\
		\ket{0\pm1}\ket{2}\\
		\ket{1\pm2}\ket{0}\\
		\ket{2}\ket{1 \pm 2}\\
		\ket{1}\ket{1}
	\end{array}
\end{equation}

These are exactly the same nine pure product states as written in (\ref{BennettPB1999}) and are therefore not locally distinguishable.

If the outcome is 2, the postmeasurement states would be

\begin{equation}
	\label{S2EB01|10|12}
	\begin{array}{c}
		\ket{0}\ket{01 \pm 12}\\
		\ket{0\pm1}\ket{10}\\
		\ket{1\pm2}\ket{01}\\
		\ket{2}\ket{12 \pm 10}\\
		\ket{1}\ket{12}
	\end{array}
\end{equation}

Now we can rewrite the above set of states in $\{\ket{0},\ket{1},\ket{2}\}$ basis in replacing of $\{\ket{01},\ket{12},\ket{10}\}$ as similar to (\ref{BennettPB1999}) and hence locally indistinguishable.	

%

\begin{thebibliography}{100}
    \bibitem{BennettPB1999}Charles H. Bennett, David P. DiVincenzo, Christopher A. Fuchs, Tal Mor, Eric Rains, Peter W. Shor, John A. Smolin, and William K. Wootters, Quantum nonlocality without entanglement, Phys. Rev. A 59, 1070(1999)
    \bibitem{BennettUPB1999}C. H. Bennett, D. P. DiVincenzo, T. Mor, P. W. Shor, J. A. Smolin, and B. M. Terhal, Unextendible Product Bases and Bound Entanglement, Phys. Rev. Lett. 82, 5385 (1999).
    \bibitem{bennett1996}C. H. Bennett, D. P. DiVincenzo, J. A. Smolin, and W. K. Wootters, Mixed-state entanglement and quantum error correction, Phys. Rev. A 54, 3824 (1996).
    \bibitem{popescu2001}H.-K. Lo and S. Popescu, Concentrating entanglement by local actions: Beyond mean values, Phys. Rev. A 63, 022301 (2001).
    \bibitem{xin2008}Y. Xin and R. Duan, Local distinguishability of orthogonal $2\otimes3$ pure states, Phys. Rev. A 77, 012315 (2008).
    \bibitem{Walgate2000} J. Walgate, A. J. Short, L. Hardy, and V. Vedral, Local Distinguishability of Multipartite Orthogonal Quantum States, Phys. Rev. Lett. 85, 4972 (2000).
    \bibitem{Virmani} S. Virmani, M. F. Sacchi, M. B. Plenio, and D. Markham, Optimal local discrimination of two multipartite pure states, Phys. Lett. A 288, 62 (2001).
    \bibitem{Ghosh2001} S. Ghosh, G. Kar, A. Roy, A. Sen(De), and U. Sen, Distinguishability of Bell States, Phys. Rev. Lett. 87, 277902 (2001).
    \bibitem{Groisman} B. Groisman and L. Vaidman, Nonlocal variables with product state eigenstates, J. Phys. A: Math. Gen. 34, 6881 (2001).
    \bibitem{Walgate2002} J. Walgate and L. Hardy, Nonlocality, Asymmetry, and Distinguishing Bipartite States, Phys. Rev. Lett. 89, 147901 (2002).
    \bibitem{Divincinzo} D. P. DiVincenzo, T. Mor, P. W. Shor, J. A. Smolin, and B. M. Terhal, Unextendible product bases, uncompletable product bases and bound entanglement, Commun. Math. Phys. 238, 379 (2003).
    \bibitem{Horodecki2003} M. Horodecki, A. Sen(De), U. Sen, and K. Horodecki, Local Indistinguishability: More Nonlocality with Less Entanglement, Phys. Rev. Lett. 90, 047902 (2003).
    \bibitem{Fan2004} H. Fan, Distinguishability and Indistinguishability by Local Operations and Classical Communication, Phys. Rev. Lett. 92, 177905 (2004).
    \bibitem{Ghosh2004} S. Ghosh, G. Kar, A. Roy, and D. Sarkar, Distinguishability of maximally entangled states, Phys. Rev. A 70, 022304 (2004).
    \bibitem{Nathanson2005} M. Nathanson, Distinguishing bipartite orthogonal states by LOCC: Best and worst cases, J. Math. Phys. 46, 062103 (2005).
    \bibitem{Watrous2005} J. Watrous, Bipartite Subspaces Having No Bases Distinguishable by Local Operations and Classical Communication, Phys. Rev. Lett. 95, 080505 (2005).
    \bibitem{Niset2006} J. Niset and N. J. Cerf, Multipartite nonlocality without entanglement in many dimensions, Phys. Rev. A 74, 052103 (2006).
    \bibitem{Ye2007} M.-Y. Ye, W. Jiang, P.-X. Chen, Y.-S. Zhang, Z.-W. Zhou, and G.-C. Guo, Local distinguishability of orthogonal quantum states and generators of SU(N), Phys. Rev. A 76, 032329 (2007).
    \bibitem{Fan2007} H. Fan, Distinguishing bipartite states by local operations and classical communication, Phys. Rev. A 75, 014305 (2007).
    \bibitem{Runyo2007} R. Duan, Y. Feng, Z. Ji, and M. Ying, Distinguishing Arbitrary Multipartite Basis Unambiguously Using Local Operations and Classical Communication, Phys. Rev. Lett. 98, 230502 (2007).
    \bibitem{somsubhro2009} S. Bandyopadhyay and J. Walgate, Local distinguishability of any three quantum states, J. Phys. A: Math. Theor. 42, 072002 (2009).
    \bibitem{Feng2009} Y. Feng and Y.-Y. Shi, Characterizing locally indistinguishable orthogonal product states, IEEE Trans. Inf. Theory 55, 2799 (2009).
    \bibitem{Runyo2010} R. Duan, Y. Xin, and M. Ying, Locally indistinguishable subspaces spanned by three-qubit unextendible product bases, Phys. Rev. A 81, 032329 (2010).
    \bibitem{Yu2012} N. Yu, R. Duan, and M. Ying, Four Locally Indistinguishable Ququad-Ququad Orthogonal Maximally Entangled States, Phys. Rev. Lett. 109, 020506 (2012).
    \bibitem{Yang2013} Y.-H. Yang, F. Gao, G.-J. Tian, T.-Q. Cao, and Q.-Y.Wen, Local distinguishability of orthogonal quantum states in a 2$\otimes$2$\otimes$2 system, Phys. Rev. A 88, 024301 (2013).
    \bibitem{Zhang2014} Z.-C. Zhang, F. Gao, G.-J. Tian, T.-Q. Cao, and Q.-Y. Wen, Nonlocality of orthogonal product basis quantum states, Phys. Rev. A 90, 022313 (2014).
    \bibitem{somsubhro2009(1)} S. Bandyopadhyay, G. Brassard, S. Kimmel, and W. K. Wootters, Entanglement cost of nonlocal measurements, Phys. Rev. A 80, 012313 (2009).
    \bibitem{somsubhro2010} S. Bandyopadhyay, R. Rahaman, and W. K. Wootters, Entanglement cost of two-qubit orthogonal measurements, J. Phys. A: Math. Theor. 43, 455303 (2010).
    \bibitem{yu2014} N. Yu, R. Duan, and M. Ying, Distinguishability of quantum states by positive operator-valued measures with positive partial
    transpose, IEEE Trans. Inf. Theory 60, 2069 (2014).
    \bibitem{somsubhro2014} S. Bandyopadhyay, A. Cosentino, N. Johnston, V. Russo, J. Watrous, and N. Yu, Limitations on separable measurements by convex optimization, IEEE Trans. Inf. Theory 61, 3593 (2014).
    \bibitem{somsubhro2016} S. Bandyopadhyay, S. Halder, and M. Nathanson, Entanglement as a resource for local state discrimination in multipartite systems, Phys. Rev. A 94, 022311 (2016).
    \bibitem{Zhang2015} Z.-C. Zhang, F. Gao, S.-J. Qin, Y.-H. Yang, and Q.-Y. Wen, Nonlocality of orthogonal product states, Phys. Rev. A 92, 012332 (2015).
    \bibitem{Wang2015} Y.-L. Wang, M.-S. Li, Z.-J. Zheng, and S.-M. Fei, Nonlocality of orthogonal product-basis quantum states, Phys. Rev. A 92, 032313 (2015).
    \bibitem{Chen2015} J. Chen and N. Johnston, The minimum size of unextendible product bases in the bipartite case (and some multipartite cases), Commun. Math. Phys. 333, 351 (2015).
    \bibitem{Yang2015} Y.-H. Yang, F. Gao, G.-B. Xu, H.-J. Zuo, Z.-C. Zhang, and Q.-Y.Wen, Characterizing unextendible product bases in qutritququad system, Sci. Rep. 5, 11963 (2015).
    \bibitem{Zhang2016} Z.-C. Zhang, F. Gao, Y. Cao, S.-J. Qin, and Q.-Y. Wen, Local indistinguishability of orthogonal product states, Phys. Rev. A 93, 012314 (2016).
    \bibitem{Xu2016(2)} G.-B. Xu, Q.-Y. Wen, S.-J. Qin, Y.-H. Yang, and F. Gao, Quantum nonlocality of multipartite orthogonal product states, Phys. Rev. A 93, 032341 (2016).
    \bibitem{Zhang2016(1)} X. Zhang, X. Tan, J. Weng, and Y. Li, LOCC indistinguishable orthogonal product quantum states, Sci. Rep. 6, 28864 (2016).
    \bibitem{Xu2016(1)} G.-B. Xu, Y.-H. Yang, Q.-Y.Wen, S.-J. Qin, and F. Gao, Locally indistinguishable orthogonal product bases in arbitrary bipartite quantum system, Sci. Rep. 6, 31048 (2016).
    \bibitem{bhunia2023} A. Bhunia, I. Biswas, I. Chattopadhyay and D. Sarkar, More assistance of entanglement, less rounds of classical communication, J. Phys. A: Math. Theor. 56 (2023) 365303.
\bibitem{bhunia2024}  A. Bhunia, S. Bera, I. Biswas, I. Chattopadhyay, and D. Sarkar, Strong quantum nonlocality: Unextendible biseparability beyond unextendible product basis, Phys. Rev. A. 109, 052211 (2024).
    \bibitem{biswas2023}  I. Biswas, A. Bhunia, I. Chattopadhyay and D. Sarkar, Entangled state distillation from single copy mixed states beyond LOCC, Phys. Lett. A, 459, 128610 (2023).
    \bibitem{Halder2019strong nonlocality} S. Halder, M. Banik, S. Agrawal, and S. Bandyopadhyay, Strong Quantum Nonlocality without Entanglement, Phys. Rev. Lett. 122, 040403 (2019).
    \bibitem{terhaldatahiding} D.P. DiVincenzo, P. Hayden, and B.M. Terhal, Hiding Quantum Data. Foundations of Physics 33, 1629-1647 (2003).
    \bibitem{haydendatahiding}  Patrick Hayden, Debbie Leung and Graeme Smith, Multiparty data hiding of quantum information, Phys. Rev. A 71, 062339 (2005)
    \bibitem{winterdatahiding} L. Lami, C. Palazuelos, and A. Winter, Ultimate Data Hiding in Quantum Mechanics and Beyond. Commun. Math. Phys. 361, 661-708 (2018).
    \bibitem{wehner2020} Victoria Lipinska, Gláucia Murta, Jérémy Ribeiro, and Stephanie Wehner, Verifiable hybrid secret sharing with few qubits, Phys. Rev. A 101, 032332 (2020). 
    \bibitem{chaves2020}  M. G. M. Moreno, Samuraí Brito, Ranieri V. Nery and Rafael Chaves, Device-independent secret sharing and a stronger form of Bell nonlocality, Phys. Rev. A 101, 052339 (2020).
    \bibitem{lamidatahiding}  Ludovico Lami, Quantum data hiding with continuous-variable systems, Phys. Rev. A 104, 052428 (2021).
 \bibitem{plahiding} I. Chattopadhyay and D. Sarkar, Phys. Lett. A 365, 273 (2007).
       \bibitem{Halder2019peres set} S. Halder, M. Banik, and S. Ghosh, Family of bound entangled states on the boundary of the Peres set, Phys. Rev. A 99, 062329 (2019).
    \bibitem{Xzhang2017} X. Zhang, J. Weng, X. Tan, and W. Luo, Indistinguishability of pure orthogonal product states by LOCC, Quantum Inf. Process. 16, 168 (2017).
    \bibitem{Xu2017} G.-B. Xu, Q.-Y. Wen, F. Gao, S.-J. Qin, and H.-J. Zuo, Local indistinguishability of multipartite orthogonal product bases, Quantum Inf. Process. 16, 276 (2017).
    \bibitem{Wang2017} Y.-L. Wang, M.-S. Li, Z.-J. Zheng, and S.-M. Fei, The local indistinguishability of multipartite product states, Quantum Inf. Process. 16, 5 (2017).
    \bibitem{Cohen2008} S. M. Cohen, Understanding entanglement as resource: Locally distinguishing unextendible product bases, Phys. Rev. A 77, 012304 (2008).
    \bibitem{Zhang2019} Z-C. Zhang, X. Zhang, Strong quantum nonlocality in multipartite quantum systems, Phys. Rev. A 99, 062108 (2019).
    \bibitem{somsubhro2018} S. Bandyopadhyay, S. Halder, and M. Nathanson, Optimal resource states for local state discrimination, Phys. Rev. A 97, 022314 (2018).
    \bibitem{zhang2018} Z.-C. Zhang, Y.-Q. Song, T.-T. Song, F. Gao, S.-J. Qin, and Q.-Y. Wen, Local distinguishability of orthogonal quantum states with multiple copies of 2$\otimes$2 maximally entangled states, Phys. Rev. A 97, 022334 (2018).
    \bibitem{Halder2018} S. Halder, Several nonlocal sets of multipartite pure orthogonal product states, Phys. Rev. A 98, 022303 (2018).
    \bibitem{Yuan2020} P. Yuan, G. Tian and X. Sun, Strong quantum nonlocality without entanglement in multipartite quantum systems, Phys. Rev. A 102, 042228 (2020).
    \bibitem{Rout2019} S. Rout, A. G. Maity, A. Mukherjee, S. Halder and M. Banik, Genuinely nonlocal product bases: Classification and entanglement-assisted discrimination, Phys. Rev. A 100, 032321 (2019).
    \bibitem{bhunia2022} A. Bhunia, I. Chattopadhyay and Debasis Sarkar, Nonlocality without entanglement: an acyclic configuration, Quantum Inf. Process. 21, 169 (2022).
    \bibitem{bhunia2020} A. Bhunia, I. Chattopadhyay and Debasis Sarkar, Nonlocality of tripartite orthogonal product states, Quantum Inf. Process. 20, 45 (2021).
    \bibitem{agrwal2019} S. Agrawal, S. Halder, and M. Banik, Genuinely entangled subspace with all-encompassing distillable entanglement across every bipartition, Phys. Rev. A 99, 032335 (2019).
    \bibitem{Cohen2007} S. M. Cohen, Phys. Rev. A 75, 052313 (2007).
    \bibitem{Xin2008} Y. Xin, R. Duan, Local distinguishability of orthogonal $2\otimes3$ pure states, Phys. Rev. A 77, 012315 (2008).
    \bibitem{BandyopadhyayHalderActivation2021} Somshubhro Bandyopadhyay and Saronath Halder, Genuine activation of nonlocality: From locally available to locally hidden information, Phys. Rev. A 104, L050201 (2021).
    \bibitem{Li2022} Mao-Sheng Li, Zhu-Jun Zheng, Genuine hidden nonlocality without entanglement:from the perspective of local discrimination,  New J. Phys. 24 043036 (2022).
    \bibitem{GhoshStrongActivation2022} Subhendu B. Ghosh, Tathagata Gupta, Ardra A. V., Anandamay Das Bhowmik, Sutapa Saha, Tamal Guha, and Amit Mukherjee, Activating strong nonlocality from local sets: An elimination paradigm, Phys. Rev. A 106, L010202 (2022).
	\bibitem{GuptaGhoshHierarchicalActivation2023}Tathagata Gupta, Subhendu B. Ghosh, Ardra A V, Anandamay Das Bhowmik, Sutapa Saha, Tamal Guha, Ramij Rahaman, and Amit Mukherjee, Hierarchical activation of quantum nonlocality: Stronger than local indistinguishability, Phys. Rev. A 107, 052418 (2023).
\bibitem{ShiLiHuDoExist2022} Fei Shi, Mao-Sheng Li, Mengyao Hu, Lin Chen, Man-Hong Yung, Yan-Ling Wang, and Xiande Zhang, Strongly nonlocal unextendible product bases do exist, q-2022-01-05-619 (2022).

\bibitem{GhoshalChoudharySen2023} Ahana Ghoshal, Swati Choudhary, Ujjwal Sen, All multipartite entanglements are quantum coherences in locally distinguishable bases, arXiv:2304.05249 (2023).

\end{thebibliography}
\end{document}